\documentclass[12pt,a4paper]{article}

\usepackage{amssymb}
\usepackage{amsmath}
\usepackage{amsthm}
\usepackage{bm}

\theoremstyle{plain}
\newtheorem{theorem}{Theorem}[section]
\newtheorem{proposition}[theorem]{Proposition}
\newtheorem{lemma}[theorem]{Lemma}
\newtheorem{corollary}[theorem]{Corollary}

\theoremstyle{definition}
\newtheorem{definition}[theorem]{Definition}

\newtheorem{assumption}[theorem]{Assumption}

\theoremstyle{remark}
\newtheorem{remark}[theorem]{Remark}
\newtheorem{remarks}[theorem]{Remarks}
\newtheorem{example}[theorem]{Example}
\newtheorem{examples}[theorem]{Examples}

\numberwithin{equation}{section}

\DeclareMathOperator{\real}{Re}
\DeclareMathOperator{\mysign}{sign}

\begin{document}
\allowdisplaybreaks
\title{A complete classification of threshold properties for one-dimensional discrete Schr\"{o}dinger
operators
}
\author{Kenichi {\scshape Ito}\footnote{Graduate School of Pure and Applied Sciences, University of Tsukuba,
1-1-1 Tennodai, Tsukuba Ibaraki, 305-8571 Japan. 
E-mail: \texttt{ito-ken@math.tsukuba.ac.jp}. 
Partially supported by JSPS Wakate (B) 21740090 (2009--2012), 25800073 (2013--2016). Partially supported by the Danish Council for Independent Research $|$ Natural Sciences, Grants 09-065927 and 11-106598}
\and
Arne {\scshape Jensen}\footnote{Department of Mathematical Sciences,
Aalborg University, Fredrik Bajers Vej 7G, DK-9220 Aalborg \O{}, Denmark.
E-mail: \texttt{matarne@math.aau.dk}. Partially supported by the Danish Council for Independent Research $|$ Natural Sciences, Grants 09-065927 and 11-106598}
}
\date{}
\maketitle

\begin{abstract}
We consider the discrete one-dimensional Schr\"{o}dinger operator $H=H_0+V$, where $(H_0x)[n]=-(x[n+1]+x[n-1]-2x[n])$ and $V$ is a self-adjoint operator on $\ell^2(\mathbb{Z})$ with a decay property given by $V$ extending to a compact operator from $\ell^{\infty,-\beta}(\mathbb{Z})$ to $\ell^{1,\beta}(\mathbb{Z})$ for some $\beta\geq1$. We give a complete description of the solutions to $Hx=0$, and $Hx=4x$, $x\in\ell^{\infty,-\beta}(\mathbb{Z})$. Using this description we give asymptotic expansions of the resolvent of $H$ at the two thresholds $0$ and $4$. One of the main results is a precise correspondence between the solutions to $Hx=0$ and the leading coefficients in the asymptotic expansion of the resolvent around $0$. 
For the resolvent expansion we implement the expansion scheme of Jensen-Nenciu \cite{JN0, JN1} in the full generality.
\end{abstract}

\tableofcontents
\section{Introduction}
In this paper we investigate the threshold properties of the resolvent of a one-dimensional discrete Schr\"{o}dinger operator with a large class of interactions. We do not discuss the general spectral properties of this class of operators, although we note that it is known that the absolutely continuous spectrum equals the interval $[0,4]$ and the singular continuous spectrum is absent under some additional assumptions on the interaction, see Remark~\ref{absence}. The eigenvalues all have finite multiplicities and there are no accumulation points for the eigenvalues. For the thresholds the last result is a consequence of the results obtained here.

In Section~\ref{sect11} of this introduction we describe the setting and define the class of interactions. We then state the main results in a simplified form. The complete statements are rather complicated and require a substantial amount of preparation. In Section~\ref{sect12} we outline the strategy of proof. In Section~\ref{sect13} we comment on the literature.

In Section~\ref{12.11.8.1.0} we prepare the results needed to obtain the asymptotic resolvent expansions. In Section~\ref{12.11.8.2.58} we describe an (iterated) inversion procedure and define some intermediate operators. Section~\ref{12.12.19.2.5} is devoted to a detailed analysis of the (generalized) eigenspaces associated with the thresholds, and an analysis of the intermediate operators, needed in the statement of the full results. Our full results on the asymptotic expansion of the resolvent are given in detail in Section~\ref{12.12.17.20.56}. The Appendices contain results on the threshold $4$ and a number of explicit examples of interactions.

\subsection{Setting and overview of results}\label{sect11}

We introduce the setting of the paper,
and  give a quick and self-contained overview of the main results
in a slightly weaker form.
The refined versions of the results
require long preparations.
They are mostly gathered in Section~\ref{12.12.17.20.56}.

Let $H_0$ be the positive one-dimensional discrete Laplacian, i.e., 
for any sequence $x\colon \mathbb Z\to \mathbb C$ we define
$H_0x\colon \mathbb Z\to \mathbb C$ by
\begin{equation*}
(H_0x)[n]=-(x[n+1]+x[n-1]-2x[n]).
\end{equation*}
The restriction of $H_0$ to the Hilbert space
\begin{equation*}
\mathcal H=
\ell^2(\mathbb{Z})=\bigl\{x\colon \mathbb{Z}\to \mathbb{C};\ 
\|x\|^2=\sum_{n\in \mathbb{Z}} |x[n]|^2<\infty\bigr\}
\end{equation*}
defines a bounded self-adjoint operator.
We call it the \textit{free Schr\"odinger operator}, and denote it by the same notation $H_0$.
The operator $H_0$ has an explicit spectral representation employing the Fourier series.
As will be seen in Section~\ref{12.11.8.1.0},
we have the spectrum:
\begin{equation*}
\sigma(H_0)=\sigma_{\mathrm{ac}}(H_0)=[0,4],
\end{equation*}
and the points $\lambda=0,4$ are the thresholds.
It is not difficult to compute explicit asymptotic expansions of
the \textit{free resolvent} $R_0(z)=(H_0-z)^{-1}$ around these thresholds.

The purpose of the present paper
is to provide asymptotic expansions of $R(z)=(H-z)^{-1}$ for perturbed $H=H_0+V$ 
around these thresholds, and investigate the relation between
the coefficients and the generalized eigenspaces.
As we will see in the Appendix~\ref{13.3.23.11.49},
the expansion around the threshold $\lambda=4$ 
is reduced to the one around $\lambda=0$,
whence in the present paper we shall consider only  $\lambda=0$ in detail.

Let us first fix our class of perturbations.
We introduce
for $s\in \mathbb R$
\begin{align*}
\mathcal L^s&=
\ell^{1,s}(\mathbb{Z})=\bigl\{x\colon \mathbb{Z}\to \mathbb{C};\ 
\|x\|_{1,s}=\sum_{n\in\mathbb{Z}}(1+n^2)^{s/2}|x[n]|<\infty\bigr\},\\
(\mathcal L^s)^*
&=\ell^{\infty,-s}(\mathbb{Z})
=\bigl\{x\colon \mathbb{Z}\to \mathbb{C};\ 
\|x\|_{\infty,-s}=\sup_{n\in\mathbb{Z}}(1+n^2)^{-s/2}|x[n]|<\infty\bigr\}.
\end{align*}
The superscript $s$ is dropped when $s=1$:
$\mathcal L=\mathcal L^1$, $\mathcal L^*=(\mathcal L^1)^*$.
We denote the set of all bounded operators 
from a general Banach space $\mathcal K$ to another 
$\mathcal K'$ by $\mathcal B(\mathcal K,\mathcal K')$,
and abbreviate $\mathcal B(\mathcal K)=\mathcal B(\mathcal K,\mathcal K)$
and ${\mathcal B}^s=\mathcal B(\mathcal L^s,(\mathcal L^s)^*)$.
We replace $\mathcal B$ by $\mathcal C$ when considering 
those for  compact operators.

\begin{assumption}\label{12.11.8.1.9}
Let $V\in {\mathcal B}(\mathcal H)$ be self-adjoint,
and assume that there exist a real number $\beta\ge 1$,
a Hilbert space ${\mathcal K}$,
an injective operator $v\in \mathcal B({\mathcal K},\mathcal L^\beta)
\cap \mathcal C({\mathcal K},\mathcal L)$
and a self-adjoint unitary operator
$U\in \mathcal B({\mathcal K})$ such that 
\begin{equation*}
V=vUv^*\in \mathcal B((\mathcal L^{\beta})^*,\mathcal L^\beta)
\cap 
\mathcal C(\mathcal L^*,\mathcal L).
\end{equation*}
\end{assumption}
\begin{remark}
\label{absence}
Using the results in \cite{BMS} one can show that under the above assumption with $\beta\geq2$ the singular continuous spectrum of $H=H_0+V$ is empty and $\sigma_{\rm ac}(H)=[0,4]$. Eigenvalues of $H$ have finite multiplicity and have no accumulation point. This last result is a consequence of the resolvent expansions obtained here. We expect that these results hold under the weaker assumption $\beta\geq1$.
\end{remark}

The possible range of $\beta\ge 1$ is subject to further restriction
according to how high an order we require for the expansions,
and hence
we shall mention an appropriate range of $\beta$ each time it is referred to.
Though formulated abstractly, Assumption~\ref{12.11.8.1.9} 
includes plenty of examples and, in particular,
allows non-local potentials. See Section~\ref{localV} for examples of potentials
satisfying Assumption~\ref{12.11.8.1.9}.
We also note that our class of interactions is additive in the sense that if $V_1$ and $V_2$ satisfy Assumption~\ref{12.11.8.1.9} for some $\beta$, then $V_1+V_2$ satisfies the assumption for this $\beta$, as can be seen by a straightforward direct sum construction.

Under Assumption~\ref{12.11.8.1.9} the perturbed Schr\"odinger operator 
$H=H_0+V$ is bounded and self-adjoint on $\mathcal H$,
and we can expand the perturbed resolvent
\begin{equation*}
R(z)=(H-z)^{-1}
\end{equation*}
as follows.
For $z\in\mathbb{C}\setminus [0,\infty)$
we take the determination of the square root with 
$\mathop{\mathrm{Im}}\sqrt{z}>0$.
\begin{theorem}\label{12.12.17.20.37}
Suppose $\beta\ge 4$ in Assumption~\ref{12.11.8.1.9},
and let $N\in [-2,\beta-6]$ be any integer.
Then 
as $z\to 0$ in $\mathbb C\setminus [0,\infty)$
the resolvent $R(z)$ has the asymptotic expansion 
in the uniform topology of $\mathcal B^{N+4}$:
\begin{equation*}
R(z)=\sum_{j=-2}^Nz^{j/2} G_j+\mathcal O(z^{(N+1)/2}),\quad
G_j\in \mathcal B^{j+3},
\end{equation*}
and the coefficients $G_j$ can be computed explicitly.
\end{theorem}
\begin{remark}\label{13.3.4.21.20}
The lower bound of $\beta$ and the operator classes of $G_j$
can be refined if we know \emph{a priori} the threshold type given in Definition~\ref{13.8.4.20.7} below.
See Section~\ref{12.12.17.20.56} for these refinements.
In particular, if the threshold is regular, it suffices to assume only $\beta\ge 2$.
In the later sections we shall use the variable $\kappa$
given by (\ref{12.11.8.9.30}) below,
and define the coefficients $G_j$ with respect to the expansion in $\kappa$.
Hence, they are different by factors.
\emph{We keep the convention of} Theorem~\ref{12.12.17.20.37} \emph{only in this subsection.}
\end{remark}

We shall investigate the first few coefficients in terms of 
the \textit{(generalized) eigenfunctions}.
Define the \textit{(generalized) eigenspace} for the threshold
$\lambda=0$ and set its dimension as 
\begin{equation*}
\widetilde{\mathcal E}=\{\Psi\in (\mathcal L^{\beta})^*;\ H\Psi=0\},\quad
\widetilde d=\dim\widetilde{\mathcal E},
\end{equation*}
respectively. 
Then we can show that the eigenspace is finite-dimensional, 
and the eigenfunctions 
have special asymptotics at infinity.
Define the sequences
$\mathbf 1,\mbox{\boldmath$\sigma$}\in (\mathcal L^0)^*$
and
$\mathbf n, |\mathbf n|\in \mathcal L^*$ by 
\begin{equation*}
\mathbf{1}[n]=1,\quad
\bm{\sigma}[n]
=
\begin{cases}
\pm 1 &\text{if $\pm n> 0$,}\\
\phantom{\pm}0&\text{if $\phantom{\pm}n=0$,}
\end{cases}
\quad
\mathbf n [n]=n,\quad
|\mathbf n|[n]=|n|,
\end{equation*}
respectively.
\begin{theorem}\label{13.3.7.13.33}
Suppose $\beta\ge 1$ in Assumption~\ref{12.11.8.1.9}.
Then 
\begin{align*}
\widetilde{\mathcal E}&\subset 
\mathbb C\mathbf n\oplus\mathbb C|\mathbf n|\oplus
\mathbb C\mathbf 1\oplus\mathbb C\bm{\sigma}
\oplus
\mathcal L^{\beta-2},\quad
\widetilde d<\infty
.
\end{align*}
\end{theorem}

In the classification of the singular part of the 
resolvent expansion
the following subspaces are canonical.
We set
\begin{align}
\mathcal E &=\widetilde{\mathcal E}\cap \bigl(
\mathbb C\mathbf 1\oplus\mathbb C\bm{\sigma}
\oplus\mathcal L^{\beta-2}\bigr),\quad
d=\dim \mathcal E;\label{13.3.7.13.34}\\
E&=\widetilde{\mathcal E}\cap \mathcal L^{\beta-2},\quad
d_0=\dim E.\label{13.3.7.13.35}
\end{align}
Obviously, we have $E\subset \mathcal E\subset \widetilde{\mathcal E}$ and 
$d_0\le d\le \widetilde d\le d_0+4$, which will be refined soon.
Now let us introduce the notion of the \textit{regular} and \textit{exceptional properties} for the threshold (similar to the terminology in \cite{JK} for continuous Schr\"{o}dinger operators in dimension three):
\begin{definition}\label{13.8.4.20.7}
The threshold $\lambda=0$ is said to be 
\begin{enumerate}
\item 
a \textit{regular point}, 
if $\mathcal E= E= \{0\}$;
\item
an \textit{exceptional point of the first kind}, 
if $\mathcal E\supsetneq E= \{0\}$;
\item
an \textit{exceptional point of the second kind}, 
if 
$\mathcal E=E\supsetneq \{0\}$;
\item
an \textit{exceptional point of the third kind}, 
if 
$\mathcal E\supsetneq E\supsetneq \{0\}$.
\end{enumerate}
\end{definition}
These properties \emph{characterize and are characterized by} 
the expansion coefficients $G_{-2}$ and $G_{-1}$ of the singular part:
\begin{theorem}\label{13.3.4.21.10}
Suppose $\beta\ge 4$ in Assumption~\ref{12.11.8.1.9}.
Then $\widetilde d=d_0+2$
and there exist bases 
$\Psi_j\in E$, $j=1,\dots, d_0$,
and 
$\Psi_j\in \mathcal E/E$, $j=d_0+1,\dots, d$,
such that
\begin{align}
G_{-2}=-\sum_{j=1}^{d_0}\langle \Psi_j,\cdot\rangle \Psi_j,\quad  
G_{-1}\equiv i\sum_{j=d_0+1}^{d}\langle \Psi_j,\cdot\rangle \Psi_j
\mod \langle E,\cdot \rangle E,
\label{13.3.11.2.46}
\end{align} 
where $\langle E,\cdot \rangle E\subset \mathcal B((\mathcal L^{\beta-2})^*,
\mathcal L^{\beta-2})$ is the subspace 
spanned by the operators of the form $\langle \Psi,\cdot \rangle \Psi'$
with 
$\Psi,\Psi'\in E$.
Furthermore, one can choose $\Psi_j\in E$, $j=1,\dots, d$,
to be orthonormal, 
and hence $-G_{-2}$ is the orthogonal projection onto $E$.
\end{theorem}
\begin{remark}\label{13.3.7.2.49}
If the spaces $E$ or $\mathcal E/E$ are trivial, 
we interpret the corresponding operators 
to the right of (\ref{13.3.11.2.46}) as $0$, respectively.
A complete classification of dimensions of the eigenspaces is obtained in 
Section~\ref{12.12.19.2.2}.
The lower bound of $\beta$ can also be refined, cf.\ Section~\ref{12.12.17.20.56}.
\end{remark}

Next, we consider $G_0$. 
We note that we need at least $\beta\ge 5$ to have well-defined $G_0$,
if we do not know \emph{a priori} the threshold type. 
The coefficient $G_0$ is actually the \textit{Green operator} for $H$.
Let us denote the identity operator by $1_{\mathcal L}=\mathrm{id}_{\mathcal L}$.
\begin{theorem}\label{13.3.4.21.11}
Suppose $\beta\ge 5$ in Assumption~\ref{12.11.8.1.9}. Then 
\begin{align*}
HG_0=G_0H=1_{\mathcal L}+G_{-2}.
\end{align*}
\end{theorem}

A detail expression for $G_0$ is possible,
and the expression is affected by the eigenspaces again.
In particular, we need to take care of eigenfunctions with certain 
\textit{quasi-symmetric} asymptotics at infinity:
\begin{definition}\label{13.3.7.13.36}
Define the \textit{quasi-symmetric eigenspace} 
and set its dimension as
\begin{align*}
\widetilde{\mathcal E}_{\mathrm{qs}}&=\widetilde{\mathcal E}
\cap \bigl(
\mathbb C|\mathbf n|\oplus\mathbb C\mbox{\boldmath$\sigma$}
\oplus
\mathcal L^{\beta-2}\bigr),\quad 
\widetilde d_{\mathrm{qs}}=\dim\widetilde{\mathcal E}_{\mathrm{qs}}.
\end{align*}
\end{definition}

It follows directly from the definition that
$E\subset \widetilde{\mathcal E}_{\mathrm{qs}}$ and 
$d_0\le \widetilde d_{\mathrm{qs}}
\le d_0+2$.

In the following theorem we use the result from Section~\ref{12.12.19.1.51} that $R_0(z)=(H_0-z)^{-1}$ has an asymptotic expansion of the form
\begin{equation}
R_0(z)=\sum_{j=-1}^Nz^{j/2}G_j^0+\mathcal{O}(z^{(N+1)/2}),
\end{equation}
with the same convention as in the statement of Theorem~\ref{12.12.17.20.37}.
\begin{theorem}\label{13.3.4.21.12}
Suppose $\beta\ge 5$ in Assumption~\ref{12.11.8.1.9}. Then 
the following three conditions are equivalent:
\begin{enumerate}
\item\label{13.3.4.18.2} The quasi-symmetric eigenspace is trivial, i.e., 
$\widetilde{\mathcal E}_{\mathrm{qs}}=\{0\}$;
\item\label{13.3.4.18.3} The operator $1_{\mathcal L^*}+G_0^0V$
is invertible in $\mathcal B(\mathcal L^*)$;
\item\label{13.3.4.18.4} The operator $1_{\mathcal L}+VG_0^0$
is invertible in $\mathcal B(\mathcal L)$.
\end{enumerate}
Furthermore, in the affirmative case, $d+1\le \widetilde d\le d+2$, and 
there exist $\Psi_j\in \widetilde{\mathcal E}$, 
$j=d+1,\dots, \widetilde d$ such that 
\begin{align*}
G_0&=(1_{\mathcal L^*}+G_0^0V)^{-1}G_0^0
+\sum_{j=d+1}^{\widetilde d}\langle \Psi_j,\cdot\rangle \Psi_j\\
&=G_0^0(1_{\mathcal L}+VG_0^0)^{-1}
+\sum_{j=d+1}^{\widetilde d}\langle \Psi_j,\cdot\rangle \Psi_j.
\end{align*}
\end{theorem}
\begin{remark}
We do not assert that $\{\Psi_j\}$ is a basis of $\widetilde{\mathcal E}/\mathcal E$, 
and in general it is not.  
Even when $\widetilde{\mathcal E}_{\mathrm{qs}}$ is nontrivial,
there exist alternative modified expressions for $G_0$. 
See Section~\ref{12.12.17.20.56}.
The lower bound of $\beta$ can also be refined.
\end{remark}

We finally note that 
in the case where $V$ is multiplicative, 
further dimensional relations hold true:
$\widetilde d=2$, $d\le 1$ and $d_0=0$.
In particular, the space $E$ is always trivial 
for the multiplicative potentials,
while the dimension $d_0=\dim E$ can be an arbitrary finite number in general 
for non-local potentials.
We shall see such examples in Appendix~\ref{13.3.23.21.47}.

The relations between 
the first few coefficients and the eigenspaces 
are mediated by certain operators,
which we name $P$, $m_0$, $q_0$ and $r_0$, cf.\ 
Section~\ref{12.11.24.5.48}.
As mentioned at the beginning,
if we formulate our statement directly in terms of 
these \textit{intermediate operators},
we can have more refined statements than the ones above, 
which is one of our goals.

\subsection{Strategy}\label{sect12}
We fix the determination of the square root with 
$\mathop{\mathrm{Im}}\sqrt{z}>0$ 
for $z\in\mathbb{C}\setminus [0,\infty)$, 
as in the previous subsection,
and introduce the new variable $\kappa$:
\begin{align}
\kappa=-i\sqrt{z};\quad z=-\kappa^2.\label{12.11.8.9.30}
\end{align}
We use the two variables $z$ and $\kappa$ interchangeably, e.g., $R(z)=R(\kappa)$, without comment.
Our first reduction is given in the following proposition,
where the expansion of $R(\kappa)$ is reduced to that 
of $R_0(\kappa)$.
Define the operator
$M(\kappa)\in\mathcal B(\mathcal K)$
for $\real\kappa>0$ by
\begin{align}
M(\kappa)=U+v^*R_0(\kappa)v.\label{12.11.8.8.1}
\end{align}
Due to the decay assumption on $V$ there exists a $\kappa_0>0$ such that for $\real\kappa\in(0,\kappa_0)$ we have that $z=-\kappa^2$ is in the resolvent set of $H$.
\begin{proposition}\label{12.11.13.19.5}
For any $0<\real\kappa<\kappa_0$ the operator $M(\kappa)$
is invertible in $\mathcal B(\mathcal K)$,
and 
\begin{align}
&R(\kappa)=R_0(\kappa)
-R_0(\kappa)vM(\kappa)^{-1}v^*R_0(\kappa),\label{12.11.8.4.1}\\
&M(\kappa)^{-1}=U-Uv^*R(\kappa)vU.\label{12.11.13.19.3}
\end{align}
\end{proposition}
\begin{proof}
By the definition (\ref{12.11.8.8.1})
and Assumption~\ref{12.11.8.1.9}
it is straightforward to see that
\begin{align*}
M(\kappa)\bigl[U-Uv^*R(\kappa)vU\bigr]
=\bigl[U-Uv^*R(\kappa)vU\bigr]M(\kappa)
=1_{{\mathcal K}}.
\end{align*}
Hence (\ref{12.11.13.19.3}) is verified.
As for (\ref{12.11.8.4.1}), we use the identities
\begin{align*}
R(\kappa)=R_0(\kappa)-R(\kappa)VR_0(\kappa)
=R_0(\kappa)-R_0(\kappa)VR(\kappa)
\end{align*}
to obtain 
\begin{align*}
R(\kappa)=R_0(\kappa)-R_0(\kappa)
v\bigl[U-Uv^*R(\kappa)vU\bigr]v^*R_0(\kappa).
\end{align*}
Then it suffices to substitute (\ref{12.11.13.19.3}) into the above equation.
\end{proof}
The expansion of $R_0(\kappa)$ is straightforward and
will be given explicitly in Section~\ref{12.12.19.1.51}.
Then by (\ref{12.11.8.8.1}) we can also expand $M(\kappa)$,
cf.\ Section~\ref{12.11.8.3.0}.
Hence by (\ref{12.11.8.4.1})
it suffices to have an expansion scheme
that expands $M(\kappa)^{-1}$ in terms of the expansion 
coefficients of $M(\kappa)$.
We will actually borrow the the scheme from \cite{JN0,JN1},
however, we emphasize that we will implement it under the most general assumption 
on perturbations.
The first step of the scheme will be given in Section~\ref{12.12.19.1.54}.
The full expansion scheme consists of iterated
applications of the inversion formula of Section~\ref{12.12.19.1.54}, 
and the outline will be given 
in Section~\ref{12.11.24.5.48}, proving 
that the iteration stops after a finite number of steps, 
as well as introducing the intermediate 
operators $P$, $m_0$, $q_0$ and $r_0$.
We will investigate the relation between these intermediate operators
and the eigenspaces in Section~\ref{12.12.19.2.5}.
An alternative proof of finiteness of iteration is also given here.
The main refined results will be stated and proved 
in Section~\ref{12.12.17.20.56}.

\subsection{Comments on the literature}\label{sect13}
We are not aware of any previous complete treatments of the resolvent expansion problem  or the associated complete classification of the threshold (generalized) eigenspaces, for the general class of interactions considered here. 

For the case of discrete one-dimensional Schr\"{o}dinger operators with multiplicative potentials there are results on the threshold expansions of the resolvent in the generic case in \cite{PS} and in the general case in \cite{C}.

The types of resolvent expansions considered here were first obtained for continuous Schr\"{o}dinger operators in \cite{rauch,JK,M}.

\section{Expansions of $R_0(\kappa)$ and $M(\kappa)$}\label{12.11.8.1.0}

\subsection{Expansion of $R_0(\kappa)$}\label{12.12.19.1.51}
For $x\in \mathcal H$ and $f\in L^2(\mathbb T)$ we use the following conventions for the Fourier transform 
${\mathcal F}\colon \mathcal H\to L^2(\mathbb T)$,
$\mathbb T=\mathbb R/2\pi \mathbb Z$.
\begin{align*}
({\mathcal F}x)(\theta)=(2\pi)^{-1/2}\sum_{n\in\mathbb{Z}}e^{-in\theta}x[n],\quad
({\mathcal F}^{-1}f)[n]=(2\pi)^{-1/2}\int_{\mathbb T} 
e^{in\theta}f(\theta)\,d\theta.
\end{align*}
Then we have 
\begin{align*}
{\mathcal F}(H_0x)(\theta)=(2-2\cos\theta)({\mathcal F}x)(\theta)
=\bigl(4\sin^2\tfrac{\theta}{2}\bigr)({\mathcal F}x)(\theta),
\end{align*}
which shows that the spectrum of $H_0$ is purely absolutely continuous and equals $[0,4]$.
It follows that, if we denote $R_0(z)=(H_0-z)^{-1}$, then
\begin{align}
{\mathcal F}(R_0(z)x)(\theta)
=\frac{({\mathcal F}x)(\theta)}{4\sin^2(\theta/2)-z},\quad z\in \mathbb{C}\setminus [0,4].
\label{12.12.27.12.55}
\end{align}
For $z\in \mathbb{C}\setminus [0,4]$ sufficiently close to $0$
let us change the variable from $z$ to $\phi$
through the correspondence 
\begin{align}
z=4\sin^2\tfrac{\phi}{2},\quad \mathop{\mathrm{Im}}\phi>0.\label{11.1.26.15.17}
\end{align}
With the correspondence (\ref{11.1.26.15.17})
the resolvent $R_0(z)$ is represented by the integral kernel,
which is a convolution with the function
\begin{align}
R_0(z;n)=\frac{ie^{i\phi |n|}}{2\sin \phi}.\label{11.1.26.16.23}
\end{align}
We now combine (\ref{11.1.26.16.23}), (\ref{11.1.26.15.17}),
and (\ref{12.11.8.9.30}) to expand $R_0(\kappa)=R_0(z)$.
\begin{proposition}\label{12.11.9.6.23}
Let $N\ge -1$ be any integer.
Then as $\kappa\to 0$ with $\real\kappa>0$, 
the resolvent $R_0(\kappa)$ has the expansion
in $\mathcal B^{N+2}$:
\begin{align}
R_0(\kappa)=\sum_{j=-1}^N\kappa^jG_j^0+{\mathcal O}(\kappa^{N+1}),
\quad G_j^0\in \mathcal B^{j+1},
\label{12.12.1.7.7}
\end{align}
and the coefficients $G_j^0$ are given explicitly 
as convolution operators with polynomials $G^0_j(n)$ of 
degree $j+1$ in $|n|$.
For instance,
\begin{align}
\begin{split}
&G^0_{-1}(n)=\tfrac{1}{2},\quad
G^0_0(n)=-\tfrac{1}{2}|n|,\quad
G^0_1(n)=\tfrac{1}{4}|n|^2-\tfrac{1}{16},\\
&G^0_2(n)=-\tfrac{1}{12}|n|^3+\tfrac{1}{12}|n|,\quad
G^0_3(n)=\tfrac{1}{48}|n|^4-\tfrac{5}{96}|n|^2+\tfrac{3}{256}.
\end{split}
\label{12.12.19.8.26}
\end{align}
\end{proposition}
\begin{proof}
The proof is straightforward and is omitted.
\end{proof}

\subsection{Expansion of $M(\kappa)$}\label{12.11.8.3.0}

By (\ref{12.11.8.8.1}) and Proposition~\ref{12.11.9.6.23}
we have the expansion of $M(\kappa)$: 
\begin{proposition}\label{12.11.27.11.18}
Suppose $\beta\ge 1$ in Assumption~\ref{12.11.8.1.9},
and let $N\in [-1,\beta-2]$ be any integer.
Then as $\kappa\to 0$ with $\real\kappa>0$, 
the operator $M(\kappa)$ has the expansion in  $\mathcal B(\mathcal K)$:
\begin{align}
M(\kappa)=\sum_{j=-1}^N\kappa^jM_j+\mathcal O(\kappa^{N+1}),
\label{11.1.27.14.40}
\end{align}
where the coefficients $M_j\in \mathcal B(\mathcal K)$ are given by 
\begin{align}
M_0=U+v^*G_0^0v,\quad
M_j=v^*G_j^0v\mbox{ for }j\neq 0.\label{11.2.1.1.40}
\end{align}
\end{proposition}
\begin{proof}
The proof is straightforward by (\ref{12.11.8.8.1}) and Proposition~\ref{12.11.9.6.23},
and is omitted.
\end{proof}

\section{Expansion scheme of Jensen-Nenciu}\label{12.11.8.2.58}
\subsection{Inversion formula}\label{12.12.19.1.54}
In Section~\ref{12.11.8.3.0} we obtained the expansion of $M(\kappa)$.
Based on this expansion we next want to expand $M(\kappa)^{-1}$. 
In this subsection we provide an
 inversion formula in a 
general setting, adapted from \cite[Corollary 2.2]{JN0}.
We consider the following condition:

\begin{assumption}\label{12.11.9.1.54}
Let ${\mathcal K}$ be a Hilbert space and 
$A(\kappa)$ a family of bounded operators on ${\mathcal K}$
with $\kappa\in D\subset \mathbb{C}\setminus \{0\}$.
Suppose that
\begin{enumerate}
\item
The set $D\subset \mathbb{C}\setminus \{0\}$ 
is invariant under the complex conjugation and accumulates at $0\in \mathbb{C}$.
\item\label{12.12.19.2.58}
For each $\kappa\in D$ the operator $A(\kappa)$ satisfies 
$A(\kappa)^*=A(\overline{\kappa})$ and has a bounded inverse $A(\kappa)^{-1}
\in \mathcal B(\mathcal K)$.
\item
As $\kappa\to 0$ in $D$, the operator $A(\kappa)$ has an expansion
in the uniform topology of the operators at ${\mathcal K}$:
\begin{align}
A(\kappa)=A_0+\kappa \widetilde A_1(\kappa);\quad  \widetilde A_1(\kappa)
={\mathcal O}(1).\label{12.11.9.3.40}
\end{align}
\item
The spectrum of $A_0$ 
does not accumulate at $0\in\mathbb{C}$ as a set.
\end{enumerate}
\end{assumption}

If the leading operator $A_0$ is invertible in $\mathcal B(\mathcal K)$, the Neumann series 
provides an inversion formula for the expansion of $A(\kappa)^{-1}$:
\begin{align*}
A(\kappa)^{-1}=\sum_{j=0}^\infty (-1)^j\kappa^j
A_0^{-1}\bigl[\widetilde A_1(\kappa)A_0^{-1}\bigr]^j.
\end{align*}
The problem is when $A_0$ has no bounded inverse, and in that case,
loosely speaking, we need the help of the lower order remainder term.
Recall that 
by Assumption~\ref{12.11.9.1.54}.\ref{12.12.19.2.58} 
we can still invert $A(\kappa)$ in $\mathcal B(\mathcal K)$ for any $\kappa\in D$.
But the inverse would have norm of order $\kappa^{-1}$ or worse
due to  the remainder.
This is where the negative exponents in the expansion
of $A(\kappa)^{-1}$
could come from.

Before stating the formula rigorously we introduce some terminology.
We define the \textit{pseudoinverse} $a^{\dagger}$ for 
a complex number $a\in \mathbb C$ by
\begin{align}
a^\dagger
=\left\{\begin{array}{ll}
0&\mbox{if $a=0$},\\
a^{-1} & \mbox{if $a\neq 0$}.
\end{array}\right.
\label{13.3.5.22.28}
\end{align} 
Let $\mathcal K'\subset \mathcal K$ be a closed subspace.
We will always identify $\mathcal B(\mathcal K')$
with its embedding in $\mathcal B(\mathcal K)$ in the standard way.
For an operator $A\in \mathcal B(\mathcal K')\subset 
\mathcal B(\mathcal K)$ 
we say that $A$ is \textit{invertible in $\mathcal B(\mathcal K')$} 
if there exists an operator $A^\dagger \in \mathcal B(\mathcal K')$
such that $A^\dagger A=AA^\dagger=1_{\mathcal K'}$,
which we identify with the orthogonal projection onto $\mathcal K'\subset \mathcal K$ as noted.
For a general self-adjoint operator $A$ on $\mathcal K$
we abuse the notation $A^\dagger$ also to denote the operator 
defined by the usual operational calculus 
for the function (\ref{13.3.5.22.28}).
The operator $A^\dagger$ for a self-adjoint operator $A$
belongs to $B(\mathcal K)$ 
if and only if the spectrum of $A$ 
does not accumulate at $0$ as a set,
and in such a case the above two $A^\dagger$ coincide.
In either case we call $A^\dagger$ the \textit{pseudoinverse} of $A$.

\begin{proposition}\label{12.11.9.3.28}
Suppose Assumption~\ref{12.11.9.1.54}.
Let $Q$ be the orthogonal projection onto $\mathop{\mathrm{Ker}} A_0$,
and define the operator $a(\kappa)\in\mathcal B(Q{\mathcal K})$
by
\begin{align}
\begin{split}
a(\kappa)&=\tfrac{1}{\kappa}
\bigl\{1_{Q\mathcal K}-Q(Q+A(\kappa))^{-1}Q\bigr\}\\
&=\sum_{j=0}^\infty (-1)^j\kappa^jQ\widetilde A_1(\kappa)
\bigl[(Q+A_0)^{-1}\widetilde A_1(\kappa)\bigr]^{j}Q.
\end{split}
\label{12.11.9.3.30}
\end{align}
Then $a(\kappa)$ is bounded in ${\mathcal B}(Q{\mathcal K})$ 
as $\kappa\to 0$ in $D$.
Moreover, for each $\kappa\in D$ sufficiently close to $0$ 
the operator $a(\kappa)$ is invertible in $\mathcal B(Q{\mathcal K})$, and 
\begin{align}
A(\kappa)^{-1}
=(Q+A(\kappa))^{-1}
+\tfrac{1}{\kappa}(Q+A(\kappa))^{-1}a(\kappa)^\dagger(Q+A(\kappa))^{-1}.
\label{12.11.9.5.30}
\end{align}
\end{proposition}

Before the proof let us describe the inversion procedure briefly.
We may assume that the leading operator $A_0$ has a nontrivial kernel.
Then by (\ref{12.11.9.5.30}) and the boundedness of $a(\kappa)$
we can see that 
negative powers of $\kappa$, at least $\kappa^{-1}$,
show up in the expansion of $A(\kappa)^{-1}$.
Assume that we have a higher-order 
expansion of $A(\kappa)$ than (\ref{12.11.9.3.40}),
i.e., $\widetilde A_1(\kappa)=A_1+\kappa \widetilde A_2(\kappa)$,
$\widetilde A_2(\kappa)={\mathcal O}(1)$.
Then the expansion of $a(\kappa)$ follows from (\ref{12.11.9.3.30}):
\begin{align*}
a(\kappa)
=a_0+\kappa\widetilde a_1(\kappa);\quad 
a_0=QA_1Q,\quad
\widetilde a_1(\kappa)={\mathcal O}(1).
\end{align*}
If the leading operator $a_0$ is invertible in $\mathcal B(Q\mathcal K)$,
then  substitution of the Neumann series for $a(\kappa)^\dagger$ 
into (\ref{12.11.9.5.30}) yields the expansion of $A(\kappa)^{-1}$.
Otherwise,
by applying Proposition~\ref{12.11.9.3.28} for $a(\kappa)$ again
we obtain the expansion of $a(\kappa)^\dagger$, 
and find that $A(\kappa)^{-1}$ has 
at least a $\kappa^{-2}$ singularity in its expansion.
We repeat this argument.
The iteration procedure stops when we encounter 
a leading operator with trivial kernel.
The asymptotic expansion 
of $A(\kappa)$ is then completed,
and the number of the iterations
gives the worst order of 
the negative powers of $\kappa$.

We note that the sequence of the kernels of the leading operators
is monotone non-increasing, however, the inversion procedure may not stop after a finite number of steps.
But in our application
this procedure always stops after a few steps, due to the self-adjointness of $H$. 
We can also prove this finiteness in a direct way.
We will see this in Sections~\ref{12.11.24.5.48} and 
\ref{12.12.19.2.2}, respectively.
\begin{proof}[Proof of Proposition~\ref{12.11.9.3.28}]
We first note that for each
$\kappa\in D$ close to $0$ the operator 
$Q+A(\kappa)$ has a bounded inverse,
and hence $a(\kappa)$ is well-defined.
Then the assertion follows if we can verify
\begin{align*}
(Q+QA(\kappa)^{-1}Q)a(\kappa)=a(\kappa)(Q+QA(\kappa)^{-1}Q)
=\tfrac1{\kappa} Q,
\end{align*}
and
\begin{align*}
1_{{\mathcal K}}&=\Bigl[(Q+A(\kappa))^{-1}
+\tfrac{1}{\kappa}(Q+A(\kappa))^{-1}a(\kappa)^\dagger(Q+A(\kappa))^{-1}\Bigr]
A(\kappa)\\
&=A(\kappa)\Bigl[(Q+A(\kappa))^{-1}
+\tfrac{1}{\kappa}(Q+A(\kappa))^{-1}a(\kappa)^\dagger(Q+A(\kappa))^{-1}\Bigr].
\end{align*}
But these are straightforward and we omit the computations.
\end{proof}

\subsection{Finiteness of inversion iteration}\label{12.11.24.5.48}

In this subsection we prove that the inversion iteration, when applied to $M(\kappa)$,  stops after a finite number of steps when $\beta\ge 4$.
We also introduce the intermediate operators $P$, $m_0$, $q_0$ 
and $r_0$, outlining the procedure of
the expansion for the resolvent $R(\kappa)$.
The detailed computations 
will be given in Section~\ref{12.12.17.20.56}.

Let $P\in\mathcal B(\mathcal K)$ be the orthogonal projection
onto the subspace spanned by 
$v^*\mathbf 1\in \mathcal K$.
We can write it as
\begin{align}
P
=\tfrac{\gamma}2\langle \Phi_1, \cdot\rangle \Phi_1;
\quad 
\gamma =2\|\Phi_1\|^{\dagger 2},\ 
\Phi_1=v^*\mathbf1.
\label{12.11.24.22.25}
\end{align} 
In order to ensure the finiteness of the iteration procedure it suffices to 
assume $\beta\ge 4$, at worst.
This lower bound of $\beta$ can be improved if we come across 
invertible $P$, $m_0$ or $q_0$ before $r_0$,
but for the moment we let $\beta\ge 4$.
Under this assumption we can write (\ref{11.1.27.14.40}) as 
\begin{align}
M(\kappa)=\tfrac{\gamma^\dagger}{\kappa}
P+M_0+\kappa M_1+\kappa^2M_2+{\mathcal O}(\kappa^3).\label{12.11.13.6.20}
\end{align}
Note that, as a consequence of (\ref{12.11.13.19.3}) and the self-adjointness
of $H$, we already know that $M(\kappa)^{-1}$ satisfies, at worst,
\begin{align}
M(\kappa)^{-1}={\mathcal O}(\kappa^{-2}).\label{12.11.24.4.17}
\end{align}
If $P$ is invertible in $\mathcal B(\mathcal K)$, i.e., $\Phi_1\neq0$ and
$\mathcal K=\mathbb C \Phi_1$, 
we can use the Neumann series to 
compute the inverse $M(\kappa)^{-1}$.
Then the inversion procedure for $M(\kappa)^{-1}$ stops, and $R(\kappa)$ can be expanded by using Proposition~\ref{12.11.13.19.5}.
Note that in this case we in fact need only $\beta\ge 1$.
Hence we may assume that $P$ is not invertible.
Apply Proposition~\ref{12.11.9.3.28} for $\kappa M(\kappa)$,
and then we obtain
\begin{align}
\begin{split}
M(\kappa)^{-1}
=\kappa(Q+\kappa M(\kappa))^{-1}
+(Q+\kappa M(\kappa))^{-1}m(\kappa)^\dagger(Q+\kappa M(\kappa))^{-1},
\end{split}
\label{12.11.11.12.35}
\end{align}
where $Q=1_{\mathcal K}-P$ and $m(\kappa)\in\mathcal B(Q\mathcal K)$ has an expansion of the form:
\begin{align}
m(\kappa)=m_0+\kappa m_1+\kappa^2m_2+\mathcal O(\kappa^3),\quad
m_j\in\mathcal B(Q\mathcal K).\label{12.11.27.8.11}
\end{align}
The explicit expressions for $m_j$ 
are listed in Lemma~\ref{12.12.19.4.4}.
Now it is reduced to expand $m(\kappa)^{\dagger}$
in $\mathcal B(\mathcal K)$.
Note that from (\ref{12.11.11.12.35}) and (\ref{12.11.24.4.17}) 
it follows 
\begin{align}
m(\kappa)^{\dagger}
=QM(\kappa)^{-1}Q-\kappa Q
={\mathcal O}(\kappa^{-2}).\label{12.11.13.21.51}
\end{align}
Similarly to the above, if $m_0$ is invertible in $\mathcal B(Q\mathcal K)$, then $m(\kappa)^{\dagger}$ can 
be computed by the Neumann series, in fact with $\beta\ge 2$. 
Hence we may assume $m_0$ is not invertible in $\mathcal B(Q\mathcal K)$,
and then by Proposition~\ref{12.11.9.3.28}
\begin{align}
m(\kappa)^{\dagger}
=(S+m(\kappa))^{\dagger}+\tfrac1{\kappa}(S+m(\kappa))^{\dagger}
q(\kappa)^{\dagger}(S+m(\kappa))^{\dagger},\label{12.11.27.6.30}
\end{align}
where $S\in\mathcal B(\mathcal K)$ is the orthogonal projection onto 
$Q\mathcal K\cap\mathop{\mathrm{Ker}}m_0$ 
and 
$q(\kappa)\in \mathcal B(S\mathcal K)$ has an expansion of the form:
\begin{align}
q(\kappa)=q_0+\kappa q_1+\mathcal O(\kappa^2),\quad
q_j\in \mathcal B(S\mathcal K).
\label{12.12.1.1.50}
\end{align}
Next, we expand $q(\kappa)^{\dagger}$ in $\mathcal B(\mathcal K)$.
Similarly to the above by (\ref{12.11.27.6.30}) and (\ref{12.11.13.21.51})
we have 
\begin{align}
q(\kappa)^{\dagger}=\kappa (Sm(\kappa)^{\dagger}S-S)
=\mathcal O(\kappa^{-1}).
\label{12.11.27.4.31}
\end{align}
Assuming $q_0$ is not invertible in $\mathcal B(S\mathcal K)$, we have by Proposition~\ref{12.11.9.3.28}
\begin{align}
q(\kappa)^{\dagger}
&=(T+q(\kappa))^{\dagger}+\tfrac1{\kappa}(T+q(\kappa))^{\dagger}
r(\kappa)^{\dagger}(T+q(\kappa))^{\dagger},
\label{12.11.27.6.43}
\end{align}
where $T\in\mathcal B(\mathcal K)$ is the orthogonal projection onto 
$S\mathcal K\cap \mathop{\mathrm{Ker}}q_0$
and $r(\kappa)\in \mathcal B(T\mathcal K)$ has an expansion of the form:
\begin{align*}
r(\kappa)=r_0+\mathcal O(\kappa),\quad r_0\in \mathcal B(T\mathcal K).
\end{align*}
By (\ref{12.11.27.6.43}) and (\ref{12.11.27.4.31}) we have 
\begin{align*}
r(\kappa)^{\dagger}=\kappa(Tq(\kappa)^{\dagger}T-T)=\mathcal O(1).
\end{align*}
This implies that $r_0$ has to be invertible
in $\mathcal B(T\mathcal K)$,
and now the iteration stops.
Note that we can also show the finiteness by directly computing $r_0$, cf.\ Corollary~\ref{12.12.28.10.55}.

Finally for  later use we collect some of coefficients of 
$m(\kappa)$, $q(\kappa)$
and $r(\kappa)$.
\begin{lemma}\label{12.12.19.4.4}
One has the explicit formulas: For $\beta\ge 1$
\begin{align*}
m_0&=QM_0Q;\\
\intertext{For $\beta\ge 2$}
m_1&=QM_1Q-QM_0(Q+\gamma P)M_0Q,\\
q_0&=Sm_1S;\\
\intertext{For $\beta\ge 3$}
m_2&=QM_2Q-QM_0(Q+\gamma P)M_1Q
-QM_1(Q+\gamma P)M_0Q\\
&\phantom{={}}+QM_0(Q+\gamma P)
M_0(Q+\gamma P)M_0Q,\\
q_1&=Sm_2S-Sm_1(S+m_0^\dagger)m_1S,\\
r_0&=Tq_1T.
\end{align*}
\end{lemma}
\begin{proof}
These coefficients are computed from the formula (\ref{12.11.9.3.30}).
\end{proof}
\begin{remark}
These formulas will be used in Section~\ref{12.12.17.20.56}.
The above lower bounds for $\beta$ are the least ones needed 
for the definitions to make sense.
In the actual expansions where these coefficients appear
we need to increase the lower bound by $1$ in order to dominate the error terms.
\end{remark}

\section{Intermediate operators}\label{12.12.19.2.5}

In Section~\ref{12.11.24.5.48} we saw that 
the singular part of the expansion of $R(\kappa)$
depends heavily on the 
operators $P$, $m_0$, $q_0$ and $r_0$.
In this section, as a preliminary step 
before the detailed computations of
the expansion, 
we investigate these operators $P$, $m_0$, $q_0$, $r_0$
and, in particular, their kernels,
cf.\ Propositions~\ref{12.12.19.6.30}--\ref{12.12.19.6.33}.

\subsection{Characterization of eigenspaces}\label{12.11.24.5.28}

\begin{lemma}\label{12.11.24.18.24}
For any $x\in\mathcal L^s$, $s\ge 1$,
the sequence $G_0^0x\in \mathcal L^*$
has the representations:
\begin{align}
(G_0^0x)[n]&=
-\tfrac{n}{2}\langle \mathbf 1,x\rangle
+\tfrac{1}{2}\langle \mathbf n,x\rangle
-\sum_{k\ge n}(k-n)x[k]\label{12.11.24.16.17}\\
&=
\tfrac{n}{2}\langle \mathbf 1,x\rangle
-\tfrac{1}{2}\langle \mathbf n,x\rangle
-\sum_{k\le n}(n-k)x[k].\label{12.11.24.16.18}
\end{align}
Moreover, $G_0^0x\in(\mathcal L^0)^*$ if and only 
if $\langle \mathbf 1,x\rangle=0$,
and $G_0^0x\in \mathcal L^{s-2}$ 
if and only if $\langle \mathbf 1,x\rangle=\langle \mathbf n,x\rangle=0$.
\end{lemma}
\begin{proof}
The identities (\ref{12.11.24.16.17}) and (\ref{12.11.24.16.18})
follow immediately from the formula
\begin{align*}
(G^0_0x)[n]
&=-\tfrac{1}{2}\sum_{k\ge n}(k-n)x[k]
-\tfrac{1}{2}\sum_{k\le n}(n-k)x[k].
\end{align*}
As for the last assertions, we note that 
\begin{align*}
\sum_{n\ge 0}(1+n^2)^{(s-2)/2}\Bigl|\sum_{k\ge n}(k-n)x[k]\Bigr|
\le C\|x\|_{1,s}<\infty.
\end{align*}
This implies
that the last summation of (\ref{12.11.24.16.17})
belongs to $\ell^{1,s-2}(\mathbb Z_+)$, and similarly for
that of (\ref{12.11.24.16.18}).
Since we have 
$\mbox{\boldmath$\sigma $},|\mathbf n|\notin \mathcal L^{s-2}$ for 
$s\ge 1$, we are done.
\end{proof}

\begin{lemma}\label{13.1.18.6.0}
The compositions $H_0G_0^0$ and $G_0^0H_0$ are 
the identity on $\mathcal L$:
\begin{align}
H_0G_0^0=G_0^0H_0=1_{\mathcal L}.
\label{13.1.18.20.17}
\end{align}
Moreover, the composition $G_0^0H_0$ is 
well-defined also on the extended space
$\mathbb C\mathbf n\oplus 
\mathbb C|\mathbf n|\oplus
\mathbb C\mathbf 1\oplus
\mathbb C\mbox{\boldmath$\sigma $}
\oplus\mathcal L$, and 
coincides with the projection $\Pi$ given as follows:
\begin{align}
G_0^0H_0=\Pi \colon 
&
\mathbb C\mathbf n\oplus 
\mathbb C|\mathbf n|\oplus
\mathbb C\mathbf 1\oplus
\mathbb C\mbox{\boldmath$\sigma$}
\oplus\mathcal L
\to 
\mathbb C|\mathbf n|\oplus
\mathbb C\mbox{\boldmath$\sigma$}\oplus
\mathcal L.\label{13.1.18.20.18}
\end{align}
\end{lemma}
\begin{remark}
Lemmas~\ref{12.11.24.18.24} and \ref{13.1.18.6.0} in particular imply that 
for $s\ge 1$
\begin{align}
\begin{split}
\mathbb C|\mathbf n|\oplus
\mathbb C\mbox{\boldmath$\sigma$}\oplus
\mathcal L^s
\subset 
G_0^0(\mathcal L^s)
\subset 
\mathbb C|\mathbf n|\oplus
\mathbb C\mbox{\boldmath$\sigma$}\oplus
\mathcal L^{s-2}.
\end{split}\label{13.1.16.1.52}
\end{align}
Thus, the sequences belonging to the image $G_0^0(\mathcal L)$
have quasi-symmetric asymptotics at infinity.
\end{remark}
\begin{proof}
For example by Lemma~\ref{12.11.24.18.24} 
we can easily verify $H_0G_0^0x=G_0^0H_0=x$ for 
$x\in\mathcal L$,
and this implies (\ref{13.1.18.20.17}).
We can also compute 
\begin{align*}
H_0\mathbf n=H_0\mathbf 1=0,\quad
G_0^0H_0|\mathbf n|=|\mathbf n|,\quad
G_0^0H_0\mbox{\boldmath$\sigma$}=\mbox{\boldmath$\sigma$},
\end{align*}
and thus the second assertion (\ref{13.1.18.20.18}) follows.
\end{proof}

We set 
\begin{align}
\begin{split}
\Psi_1^0&= \mathbf 1,\quad
\Psi_2^0=\mathbf n-\langle v^*\mathbf 1,v^*\mathbf n\rangle\|v^*\mathbf 1\|^{\dagger 2} \mathbf 1,\\
\Phi_1&=v^*\Psi_1^0,\quad
\Phi_1^*=\|\Phi_1\|^{\dagger 2}\Phi_1,\\
\Phi_2&=v^*\Psi_2^0=Qv^*\mathbf n,\quad
\Phi_2^*=\|\Phi_2\|^{\dagger 2}\Phi_2,
\end{split}\label{13.3.20.2.33}
\end{align}
where $\Phi_1$ is the same as that in (\ref{12.11.24.22.25}),
and define $\widetilde P,\widetilde Q\in \mathcal B(\mathcal K)$ by
\begin{align*}
\widetilde P
=\langle \Phi_1^*,\cdot\rangle \Phi_1
+\langle \Phi_2^*,\cdot\rangle \Phi_2,\quad 
\widetilde Q=1_{\mathcal K}-\widetilde P.
\end{align*}
Obviously, $\widetilde P$ is the orthogonal projection onto 
the subspace $\mathbb C\Phi_1\oplus \mathbb C\Phi_2\subset \mathcal K$.
We define the operators $w\in\mathcal B((\mathcal L^\beta)^*,\mathcal K)$
and $z\in \mathcal B(\mathcal K, \mathcal L^*)$ by 
\begin{align}
w=Uv^*,\quad 
z
=\langle M_0\Phi_1^*,\cdot\rangle \Psi_1^0
+\langle M_0\Phi_2^*,\cdot\rangle \Psi_2^0
-G_0^0v.
\label{12.12.27.16.4}
\end{align}
\begin{proposition}\label{12.12.27.16.6}
Suppose $\beta\ge 1$ in Assumption~\ref{12.11.8.1.9}.
Then 
\begin{align}
w\circ z|_{\mathop{\mathrm{Ker}}\widetilde QM_0}
=1_{\mathop{\mathrm{Ker}}\widetilde QM_0},\label{12.11.26.19.10}
\end{align}
and 
\begin{align}
z^{-1}(\widetilde{\mathcal E})&=\mathop{\mathrm{Ker}} \widetilde QM_0,
&
\mathop{\mathrm{Ker}}w|_{\widetilde{\mathcal E}}
&=\mathop{\mathrm{Ker}} v^*|_{\mathbb C\Psi_1^0\oplus \mathbb C\Psi_2^0},
\label{12.11.26.19.0}\\
z^{-1}(\mathcal E)
&=\mathop{\mathrm{Ker}} P\cap \mathop{\mathrm{Ker}}QM_0,&
\mathop{\mathrm{Ker}}w|_{\mathcal E}
&=\mathop{\mathrm{Ker}} v^*|_{\mathbb C\Psi_1^0},
\label{12.11.26.19.1}\\
z^{-1}(E)
&=\mathop{\mathrm{Ker}} \widetilde P\cap \mathop{\mathrm{Ker}}M_0,&
\mathop{\mathrm{Ker}}w|_{E}
&=\{0\},
\label{12.11.27.2.28}\\
z^{-1}(\widetilde{\mathcal E}_{\mathrm{qs}})
&=\mathop{\mathrm{Ker}}M_0,&
\mathop{\mathrm{Ker}}w|_{\widetilde{\mathcal E}_{\mathrm{qs}}}
&=\{0\}.
\label{12.11.27.2.29}
\end{align}
\end{proposition}
\begin{remark}
We have not yet verified the asymptotics in Theorem~\ref{13.3.7.13.33},
however, definitions (\ref{13.3.7.13.34}), (\ref{13.3.7.13.35})
and Definition~\ref{13.3.7.13.36}
themselves make sense in any case.
The above $\mathcal E$, $E$ and $\widetilde{\mathcal E}_{\mathrm{qs}}$ are understood
in this way.
\end{remark}
\begin{proof}
\textit{Step 1.}
We can compute for $\Phi\in \mathop{\mathrm{Ker}}\widetilde QM_0$, using $v^{\ast}G_0^0v=M_0-U$,
\begin{align*}
w\circ z\Phi
&=\langle M_0\Phi_1^*,\Phi\rangle U\Phi_1
+\langle M_0\Phi_2^*,\Phi\rangle U\Phi_2+U(U-M_0)\Phi\\
&=U\widetilde PM_0\Phi+\Phi-UM_0\Phi\\
&=\Phi.
\end{align*}
This implies (\ref{12.11.26.19.10}).

\smallskip
\noindent
\textit{Step 2.}
Next, we prove (\ref{12.11.26.19.0}).
For $\Phi \in \mathcal K$ by using (\ref{13.1.18.20.17}) and $v^{\ast}G_0^0v=M_0-U$,
\begin{align*}
Hz\Phi
&=\langle M_0\Phi_1^*,\Phi\rangle vU\Phi_1
+\langle M_0\Phi_2^*,\Phi\rangle vU\Phi_2-v\Phi-vUv^*G^0_0v\Phi\\
&= vU\widetilde PM_0\Phi-vUM_0\Phi\\
&=-vU\widetilde QM_0\Phi.
\end{align*}
Hence, by the injectivity of $v$ 
it follows that $z\Phi\in \widetilde{\mathcal E}$ if and only if
$\Phi\in\mathop{\mathrm{Ker}}\widetilde QM_0$,
which implies the first identity of (\ref{12.11.26.19.0}).
As for the second identity,
we note that for any $\Psi\in (\mathcal L^{\beta})^*$
we have $\Psi\in \mathop{\mathrm{Ker}}w|_{\widetilde{\mathcal E}}$
if and only if 
\begin{align*}
H_0\Psi =0,\quad
v^*\Psi=0.
\end{align*}
Since $H_0\Psi =0$ gives a difference equation of order 2,
we can rephrase it as
$\Psi\in \mathbb C\Psi_1^0\oplus \mathbb C\Psi_2^0$.
Thus 
we obtain the second identity of (\ref{12.11.26.19.0}).

\smallskip
\noindent
\textit{Step 3.}
Let us prove (\ref{12.11.26.19.1}).
Let $\Phi\in \mathcal K$.
By Lemma~\ref{12.11.24.18.24} we have two expressions for $z\Phi$:
\begin{align}
\begin{split}
z\Phi[n]
&=
\bigl[\langle \Phi_2^*,M_0\Phi\rangle
+\tfrac{1}2\langle \Phi_1,\Phi\rangle\bigr]\Psi_2^0[n]
+\bigl[\langle \Phi_1^*,M_0\Phi\rangle
-\tfrac12\langle \Phi_2,\Phi\rangle\bigr]\Psi_1^0[n]
\\&\phantom{={}}
+\sum_{k\ge n} (k-n)(v\Phi)[k]\\
&=
\bigl[\langle \Phi_2^*,M_0\Phi\rangle
-\tfrac{1}2\langle \Phi_1,\Phi\rangle\bigr]\Psi_2^0[n]
+\bigl[\langle \Phi_1^*,M_0\Phi\rangle
+\tfrac12\langle \Phi_2,\Phi\rangle\bigr]\Psi_1^0[n]
\\&\phantom{={}}
+\sum_{k\le n} (n-k)(v\Phi)[k].
\end{split}\label{b11.1.30.19.38}
\end{align}
As in the proof of Lemma~\ref{12.11.24.18.24},
the two summations in (\ref{b11.1.30.19.38})
belong to $\ell^{1,\beta-2}(\mathbb Z_\pm)$, respectively.
This fact combined with the first identity of (\ref{12.11.26.19.0})
implies that $z\Phi\in \mathcal E$
if and only if 
\begin{align*}
\Phi\in\mathop{\mathrm{Ker}}\widetilde QM_0,\quad
\langle \Phi_2^*,M_0\Phi\rangle=\langle \Phi_1,\Phi\rangle =0.
\end{align*}
Hence the first identity of (\ref{12.11.26.19.1}) is obtained.
As for the second one we can proceed as in Step 2, and 
it is almost obvious.

\smallskip
\noindent
\textit{Step 4.}
The assertions (\ref{12.11.27.2.28}) and (\ref{12.11.27.2.29})
can be shown as in Step 3, using in particular (\ref{b11.1.30.19.38}).
We omit the details.
\end{proof}

The identity (\ref{12.11.26.19.10})
combined with the first of (\ref{12.11.26.19.0}) 
implies that the restrictions
\begin{align*}
z|_{\mathop{\mathrm{Ker}}\widetilde QM_0}
\colon\mathop{\mathrm{Ker}}\widetilde QM_0\to \widetilde{\mathcal E},\quad
w|_{\widetilde{\mathcal E}}\colon \widetilde{\mathcal E}\to \mathop{\mathrm{Ker}}\widetilde QM_0
\end{align*}
are injective and surjective, respectively. 
Hence, we have a very important corollary:

\begin{corollary}\label{13.1.16.2.51}
The eigenspaces are identified as follows:
\begin{align}
\widetilde{\mathcal E}&=
z(\mathop{\mathrm{Ker}} \widetilde QM_0)\oplus
\mathop{\mathrm{Ker}} v^*|_{\mathbb C\Psi_1^0\oplus \mathbb C\Psi_2^0},
\label{12.12.19.6.56}\\
\mathcal E
&=
z(\mathop{\mathrm{Ker}} P\cap \mathop{\mathrm{Ker}}  QM_0)
\oplus
\mathop{\mathrm{Ker}} v^*|_{\mathbb C\Psi_1^0},
\label{12.12.19.6.57}\\
E&
=z(\mathop{\mathrm{Ker}} \widetilde P\cap \mathop{\mathrm{Ker}}M_0),
\label{12.12.19.6.58}\\
\widetilde{\mathcal E}_{\mathrm{qs}}&
=z(\mathop{\mathrm{Ker}}M_0).
\label{12.12.19.6.59}
\end{align}
In particular, the eigenfunctions have the asymptotics:
\begin{align}
\widetilde{\mathcal E}&\subset 
\mathbb C\mathbf n\oplus
\mathbb C|\mathbf n|\oplus
\mathbb C\mathbf 1\oplus
\mathbb C\mbox{\boldmath$\sigma$}\oplus
\mathcal L^{\beta-2}.
\label{13.3.7.13.49}
\end{align}
\end{corollary}
\begin{proof}
The isomorphisms (\ref{12.12.19.6.56})--(\ref{12.12.19.6.59}) are direct
consequences of Proposition~\ref{12.12.27.16.6}.
The asymptotics (\ref{13.3.7.13.49}) follows from isomorphism 
(\ref{12.12.19.6.56}), the definition
(\ref{12.12.27.16.4}) of $z$ and (\ref{13.1.16.1.52}).
\end{proof}

We have yet another characterization for $\widetilde{\mathcal E}_{\mathrm{qs}}$
and also the finite dimensionality of eigenspaces:
\begin{proposition}\label{13.1.16.1.59}
Suppose $\beta\ge 1$ in Assumption~\ref{12.11.8.1.9}.
Then for the compact operators 
$G_0^0V\in \mathcal C(\mathcal L^*)$,
$VG_0^0\in\mathcal C(\mathcal L)$ and 
$v^*G_0^0v\in \mathcal C(\mathcal K)$
there exists a well-defined circular sequence of isomorphisms
between the finite dimensional vector spaces:
\begin{align}
\begin{split}
\cdots&
\stackrel{-G_0^0}{\to}
\mathop{\mathrm{Ker}}(1_{\mathcal L^*}+G_0^0V)
\stackrel{Uv^*}{\to}
\mathop{\mathrm{Ker}} M_0
\stackrel{v}{\to} 
\mathop{\mathrm{Ker}}(1_{\mathcal L}+VG_0^0)\\
&\stackrel{-G_0^0}{\to}
\mathop{\mathrm{Ker}}(1_{\mathcal L^*}+G_0^0V)
\stackrel{Uv^*}{\to}
\cdots .
\end{split}\label{13.1.16.0.14}
\end{align}
In particular, 
$\widetilde d<\infty$, and 
the quasi-symmetric eigenspace is characterized by  
\begin{align}
\widetilde{\mathcal E}_{\mathrm{qs}}
=-G_0^0v\mathop{\mathrm{Ker}}M_0
=\mathop{\mathrm{Ker}}(1_{\mathcal L^*}+G_0^0V)
=-G_0^0\mathop{\mathrm{Ker}}(1_{\mathcal L}+VG_0^0)
.
\label{13.3.6.9.36}
\end{align}
\end{proposition}
\begin{proof}
The compactness of the operators $G_0^0V$, $VG_0^0$ and $M_0$
is an immediate consequence of Assumption~\ref{12.11.8.1.9},
and hence the kernels in the sequence (\ref{13.1.16.0.14})
are of finite dimensions.
It is not difficult to show that the sequence (\ref{13.1.16.0.14})
is actually well-defined and circular,
and, furthermore, that any compositions of three adjacent operators
are the identities.
This implies that all operators in (\ref{13.1.16.0.14}) are bijective.
Then noting $z=-G_0^0v$ on $\mathop{\mathrm{Ker}}M_0$, 
we obtain the characterization (\ref{13.3.6.9.36})
by Corollary~\ref{13.1.16.2.51} and (\ref{13.1.16.0.14}),
and also $\widetilde d_{\mathrm{qs}}<\infty$.
Finally, $\widetilde d\le d_0+4\le \widetilde d_{\mathrm{qs}}+4<\infty$.
Hence we are done.
\end{proof}

We finally add a remark on the dimensions of eigenspaces for 
multiplicative potentials.
With it for multiplicative potentials
we can skip a part of the following subsection.
\begin{proposition}\label{13.3.25.13.0}
Suppose that $\beta\ge 1$ in Assumption~\ref{12.11.8.1.9},
and that $V$ is \emph{multiplicative}.
Then
\begin{align*}
\widetilde d\le 2,\quad d\le 1,\quad d_0=0.
\end{align*}
\end{proposition}
\begin{proof}
For a multiplicative potential
the equation $H\Psi=0$ is a difference equation of order $2$,
and hence it is clear that $\widetilde d\le 2$.

For the remaining assertions it suffices to show that the subspace 
\begin{align*}
\widetilde{\mathcal E}^-
=\widetilde{\mathcal E}\cap 
\bigl(\mathbb C(\mathbf n-|\mathbf n|)\oplus
\mathbb C(\mathbf 1-\mbox{\boldmath$\sigma$})\oplus
\mathcal L^{\beta-2}\bigr)
\end{align*}
is actually trivial.
We let $\Psi\in \widetilde{\mathcal E}^-$ and write down the relation $\Psi=zw\Psi$
explicitly. Taking into account the asymptotics of $\Psi$ as $n\to \infty$,
we have to have 
\begin{align*}
\Psi[n]
=\sum_{k\ge n} (k-n)V[k]\Psi[k],
\end{align*}
cf.\ (\ref{b11.1.30.19.38}).
Note that by $V\in \mathcal L^\beta$ we can 
choose large $n_0\ge 0$ such that 
\begin{align*}
\sum_{n\ge n_0}n|V[n]|\le \tfrac12.
\end{align*}
Then we obtain for $n\ge n_0$
\begin{align*}
|\Psi[n]|\le \tfrac12 \sup_{k\ge n_0}|\Psi[k]|,
\end{align*}
which implies $\Psi[n]=0$ for $n\ge n_0$.
Since the equation $H\Psi=0$ is a difference equation of order $2$,
the above initial condition yields $\Psi=0$. Hence $\widetilde{\mathcal E}^-=\{0\}$, and we are done.
\end{proof}

\subsection{Eigenspaces and intermediate operators}\label{12.12.19.2.2}

In this subsection we provide a complete classification of the eigenspaces.
We can find explicit bases for them modulo $E$.

We shall investigate structures of the eigenspaces 
through the \textit{block decomposition argument} for $M_0$.
Let us decompose the action of $M_0$ into the 
blocks corresponding to the spatial decomposition:
\begin{align}
\mathcal K=\mathbb C \Phi_1\oplus S^\perp Q\mathcal K\oplus S\mathcal K,
\quad S^\perp=1_{\mathcal K}-S.
\label{13.3.9.1.16}
\end{align}

\begin{lemma}\label{13.3.9.3.47}
Suppose $\beta\ge 1$ in Assumption~\ref{12.11.8.1.9}.
Then the identities
\begin{align}
M_0^*=M_0,\quad
m_0^\dagger M_0Q=QM_0m_0^\dagger=S^\perp Q,\quad
SM_0Q=QM_0S=0
\label{13.3.20.4.5}
\end{align}
hold.
In addition, $M_0^\dagger M_0=M_0M_0^\dagger=S^\perp$
hold if and only if $\Phi_1=0$ or 
\begin{align}
 \langle \Phi_1,(M_0-M_0m_0^{\dagger}M_0)\Phi_1\rangle\neq 0.
\label{13.3.20.3.57}
\end{align}
\end{lemma}
\begin{remark}\label{13.3.20.7.9}
Stated differently, 
Lemma~\ref{13.3.9.3.47} says 
that according to the decomposition (\ref{13.3.9.1.16}) 
$M_0$ has the following $3\times 3$ block matrix 
representation
\begin{align}
M_0=
\left(\begin{array}{ccc}
?&?&?\\
?&*&0\\
?&0&0
\end{array}\right),\label{13.3.9.1.17}
\end{align}
where $*$ and $?$ mean invertible and undetermined components, respectively.
In case (\ref{13.3.9.1.16}) has trivial subspaces as components
we understand the corresponding columns and rows in 
(\ref{13.3.9.1.17}) as eliminated.
When $\Phi_1\neq 0$, the upper left $2\times 2$ submatrix of (\ref{13.3.9.1.17})
is invertible if and only if (\ref{13.3.20.3.57}) holds.
\end{remark}
\begin{proof}
The identities in (\ref{13.3.20.4.5})
are obvious from the definitions.
Next, we note that the identities 
$M_0^\dagger M_0=M_0M_0^\dagger=S^\perp$
are equivalent to the invertibility of 
$S^\perp M_0S^\perp $
in $\mathcal B(S^\perp \mathcal K)$.
If $\Phi_1=0$, then it is obvious by the defintion of $S$ that 
$S^\perp M_0S^\perp $ is invertible in $\mathcal B(S^\perp \mathcal K)$.
Hence we may consider only the case $\Phi_1\neq 0$.
Then since the operator $S^\perp M_0S^\perp $ is expressed
by the upper left $2\times 2$ submatrix of (\ref{13.3.9.1.17}),
the assertion follows as an easy consequence of  
the Gaussian elimination for matrices,
cf.\ the Schur complement. We omit the details.
\end{proof}

Let us introduce a set of important vectors in addition to
(\ref{13.3.20.2.33}): 
\begin{align*}
\Psi_3&=z\Phi_3,\quad \Phi_3=2 SM_0\Phi_1^*,\quad \Phi_3^*=\|\Phi_3\|^{\dagger 2}\Phi_3\\
\Psi_4&=z\Phi_4,\quad \Phi_4=S\Phi_2
-\langle \Phi_3^*,\Phi_2\rangle \Phi_3,\quad \Phi_4^*=\|\Phi_4\|^{\dagger 2}\Phi_4,\\
\Psi_5&=z\Phi_5,\quad \Phi_5=(1_{\mathcal K}-m_0^\dagger M_0)\Phi_1^*,\\
\Psi_6&=z\Phi_6,\quad \Phi_6=m_0^\dagger \Phi_2
+2\langle \Phi_5-2\Delta \Phi_3^*,\Phi_2\rangle \Phi_3^*
+2\langle \Phi_3^*,\Phi_2\rangle \Phi_5,
\end{align*}
where 
\begin{align*}
\Delta&=\langle \Phi_1^*,M_0\Phi_5\rangle
=\|\Phi_1\|^{\dagger 4}\langle \Phi_1,(M_0-M_0m_0^{\dagger}M_0)\Phi_1\rangle.
\end{align*}
We will see that combinations of these vectors span the eigenspaces.
The behavior of these vectors under $M_0$ is essential in our argument because it determines
the block matrix components of $M_0$.
We can compute it here:
\begin{align}
\begin{split}
&M_0\Phi_3=\tfrac12\|\Phi_3\|^2\Phi_1,\quad 
M_0\Phi_4=0,\\
&M_0\Phi_5=\Delta \Phi_1+\tfrac12\Phi_3,\quad
M_0\Phi_6=\Phi_2-\Phi_4.
\end{split}\label{13.3.9.14.47b}
\end{align}
In particular, $M_0\Phi_5\in \mathbb C\Phi_1$ and $M_0\Phi_6\in\mathbb C\Phi_2$ if
and only if $\Phi_3=0$ and $\Phi_4=0$, respectively.
For later use in Section~\ref{12.12.17.20.56} we also note that 
\begin{align}
\begin{split}
(1_{\mathcal L^*}+G_0^0V)\Psi_3
&=
\tfrac12\|\Phi_3\|^2 \Psi_1^0,\\
(1_{\mathcal L^*}+G_0^0V)\Psi_5
&=
\Delta \Psi_1^0
+\tfrac12\langle \Phi_2^*,\Phi_3\rangle \Psi_2^0
-\tfrac12G_0^0vU\widetilde Q \Phi_3,\\
(1_{\mathcal L^*}+G_0^0V)\Psi_6
&=
\Psi_2^0
-\langle \Phi_2^*,\Phi_4\rangle \Psi_2^0+G_0^0vU\widetilde Q\Phi_4.
\end{split}\label{13.3.9.14.47}
\end{align}
In particular, $(1_{\mathcal L^*}+G_0^0V)\Psi_5=\Delta \Psi_1^0$ and 
$(1_{\mathcal L^*}+G_0^0V)\Psi_6=\Psi_2^0$ if $\Phi_3=0$ and $\Phi_4=0$, respectively.
Now we can state the main propositions of this subsection.
Their proofs will be given later.

\begin{proposition}\label{12.12.19.6.30}
Suppose $\beta\ge 1$ in Assumption~\ref{12.11.8.1.9}.
Then $P$ is invertible in $\mathcal B(\mathcal K)$
if and only if $\mathcal K=\mathbb C \Phi_1$.
Furthermore, if $P$ is invertible in $\mathcal B(\mathcal K)$,
bases of the eigenspaces are given 
as in Table~\ref{13.3.6.0.26} according to the Cases
defined in Table~\ref{13.3.6.0.26a},
where the entries with parentheses are automatically determined by the assumption and 
those without.
\begin{table}[h]
\begin{center}
\begin{tabular}{|l|ccccc|}
\hline
& $\Phi_1$ & $\Phi_2$ & $\Phi_3$ & $\Phi_4$ &  $\Delta$ \\ 
\hline
Case i.& $=0$&($=0$)&($=0$)&($=0$)&($=0$)\\ 
Case i\hspace{-.05em}i.& $\neq 0$& ($=0$)&($=0$)&($=0$)&$=0$\\ 
Case i\hspace{-.05em}i\hspace{-.05em}i.& $\neq 0$& ($=0$)&($=0$)&($=0$)&$\neq 0$\\ 
\hline
\end{tabular}
\caption{Definitions of Cases (Proposition~\ref{12.12.19.6.30}).}
\label{13.3.6.0.26a}
\end{center}
\end{table}
\begin{table}[h]
\begin{center}
\begin{tabular}{|l|ccccc|}
\hline
\rule{0cm}{0.5cm}& $\widetilde{\mathcal E}$ & $ \mathcal E$  & 
  $E$& $\widetilde{\mathcal E}_{\mathrm{qs}}$ & type\\ 
\hline
Case i.& $\{\Psi_1^0,\Psi_2^0\}$& $\{\Psi_1^0\}$&$\emptyset$&$\emptyset$&exceptional I\\ 
Case i\hspace{-.05em}i.& $\{\Psi_5,\Psi_2^0\}$& $\emptyset$&$\emptyset$&$\{\Psi_5\}$&regular\\ 
Case i\hspace{-.05em}i\hspace{-.05em}i.& $\{\Psi_5,\Psi_2^0\}$& $\emptyset$&$\emptyset$&$\emptyset$&regular\\ 
\hline
\end{tabular}
\caption{Bases of the resulting eigenspaces (Proposition~\ref{12.12.19.6.30}).}
\label{13.3.6.0.26}
\end{center}
\end{table}
\end{proposition}

\begin{proposition}\label{12.12.19.6.31}
Suppose that $\beta\ge 2$ in Assumption~\ref{12.11.8.1.9},
and that $P$ is not invertible in $\mathcal B(\mathcal K)$.
Then
\begin{align}
Q\mathcal K\cap \mathop{\mathrm{Ker}}m_0
=\mathop{\mathrm{Ker}} P\cap \mathop{\mathrm{Ker}}  QM_0
\cong \mathcal E
\big/(\mathop{\mathrm{Ker}} v^*|_{\mathbb C \Psi_1^0}).
\label{12.12.19.7.0}
\end{align}
Furthermore, if $m_0$ is invertible in $\mathcal B(Q\mathcal K)$,
bases of the eigenspaces are given 
as in Table~\ref{13.3.6.0.27} according to the Cases
defined in Table~\ref{13.3.6.0.27a},
where the entries with parentheses are automatically determined by the assumption and 
those without.
\begin{table}[h]
\begin{center}
\begin{tabular}{|l|ccccc|}
\hline
& $\Phi_1$ & $\Phi_2$ & $\Phi_3$ & $\Phi_4$ &  $\Delta$ \\ 
\hline
Case i. & $=0$ & $=0$ & ($=0$) & ($=0$) & ($=0$)\\ 
Case i\hspace{-.05em}i. & $\neq 0$& $= 0$&($=0$)&($=0$)&$=0$\\ 
Case i\hspace{-.05em}i\hspace{-.05em}i. & $\neq 0$& $= 0$&($=0$)&($=0$)&$\neq 0$\\ 
Case i\hspace{-.05em}v. & $=0$& $\neq 0$&($=0$)&($=0$)&($=0$)\\ 
Case v.& $\neq 0$& $\neq 0$&($=0$)&($=0$)&$=0$\\ 
Case v\hspace{-.05em}i. & $\neq 0$& $\neq 0$&($=0$)&($=0$)&$\neq 0$\\ 
\hline
\end{tabular}
\caption{Definitions of Cases (Proposition~\ref{12.12.19.6.31}).}
\label{13.3.6.0.27a}
\end{center}
\end{table}
\begin{table}[h]
\begin{center}
\begin{tabular}{|l|ccccc|}
\hline
\rule{0cm}{0.5cm}& $\widetilde{\mathcal E}$ & $\mathcal E$ &
 $E$ &$\widetilde{\mathcal E}_{\mathrm{qs}}$ & type \\ 
\hline
Case i.& $\{\Psi_1^0,\Psi_2^0\}$& $\{\Psi_1^0\}$&$\emptyset$&$\emptyset$&exceptional I\\ 
Case i\hspace{-.05em}i.& $\{\Psi_5,\Psi_2^0\}$& $\emptyset$&$\emptyset$&$\{\Psi_5\}$&regular\\ 
Case i\hspace{-.05em}i\hspace{-.05em}i.& $\{\Psi_5,\Psi_2^0\}$& $\emptyset$&$\emptyset$&$\emptyset$&regular\\ 
Case i\hspace{-.05em}v.& $\{\Psi_1^0,\Psi_6\}$& $\{\Psi_1^0\}$&$\emptyset$&$\emptyset$&exceptional I\\ 
Case v.& $\{\Psi_5,\Psi_6\}$& $\emptyset$&$\emptyset$&$\{\Psi_5\}$&regular\\ 
Case v\hspace{-.05em}i.& $\{\Psi_5,\Psi_6\}$& $\emptyset$&$\emptyset$&$\emptyset$&regular\\ 
\hline
\end{tabular}
\caption{Bases of the resulting eigenspaces (Proposition~\ref{12.12.19.6.31}).}
\label{13.3.6.0.27}
\end{center}
\end{table}
\end{proposition}

\begin{proposition}\label{12.12.19.6.32}
Suppose that $\beta\ge 3$ in Assumption~\ref{12.11.8.1.9},
and that $P$ and $m_0$ are not invertible in 
$\mathcal B(\mathcal K)$ and $\mathcal B(Q\mathcal K)$, respectively.
Then
\begin{align}
&q_0=-\tfrac12\bigl[\langle S\Phi_2,\cdot\rangle S\Phi_2
+\langle \Phi_3,\cdot\rangle \Phi_3\bigr],
\label{12.12.19.8.41}\\
&S\mathcal K\cap \mathop{\mathrm{Ker}}q_0
=\mathop{\mathrm{Ker}} \widetilde P\cap \mathop{\mathrm{Ker}}M_0
\cong E.\label{12.12.19.8.42}
\end{align}
Furthermore, if $q_0$ is invertible in $\mathcal B(S\mathcal K)$,
bases of the eigenspaces are given 
as in Table~\ref{13.3.6.0.28} according to the Cases
defined in Table~\ref{13.3.6.0.28a},
where the entries with parentheses are automatically determined by the assumption and 
those without.
\begin{table}[h]
\begin{center}
\begin{tabular}{|l|ccccc|}
\hline
& $\Phi_1$ & $\Phi_2$ & $\Phi_3$ & $\Phi_4$ &  $\Delta$ \\ 
\hline
Case v\hspace{-.05em}i\hspace{-.05em}i.& $=0$& ($\neq 0$)&($=0$)&($\neq 0$)&($=0)$\\ 
Case v\hspace{-.05em}i\hspace{-.05em}i\hspace{-.05em}i.& ($\neq 0)$& $= 0$&$\neq 0$&$= 0$&arbitrary\\ 
Case i\hspace{-.05em}x.& $\neq 0$& ($\neq 0$)&$=0$&($\neq 0$)&$=0$\\ 
Case x.& $\neq 0$& ($\neq 0$)&$=0$&($\neq 0$)&$\neq 0$\\ 
Case x\hspace{-.05em}i.& ($\neq 0$)& $\neq 0$&$\neq 0$&$= 0$&arbitrary\\ 
Case x\hspace{-.05em}i\hspace{-.05em}i.& ($\neq 0)$& ($\neq 0$)&$\neq 0$&$\neq 0$&arbitrary\\ 
\hline
\end{tabular}
\caption{Definitions of Cases (Proposition~\ref{12.12.19.6.32}).}
\label{13.3.6.0.28a}
\end{center}
\end{table}
\begin{table}[h]
\begin{center}
\begin{tabular}{|l|ccccc|}
\hline
\rule{0cm}{0.5cm}&
$\widetilde{\mathcal E}$ & $ \mathcal E$ & $E$& $\widetilde{\mathcal E}_{\mathrm{qs}}$ & type \\ 
\hline
Case v\hspace{-.05em}i\hspace{-.05em}i.&
$\{\Psi_1^0,\Psi_4\}$& $\{\Psi_1^0,\Psi_4\}$&$\emptyset$&$\{\Psi_4\}$&exceptional I\\ 
Case v\hspace{-.05em}i\hspace{-.05em}i\hspace{-.05em}i.&
$\{\Psi_3,\Psi_2^0\}$& $\{\Psi_3\}$&$\emptyset$&$\emptyset$&exceptional I\\  
Case i\hspace{-.05em}x.&
$\{\Psi_5,\Psi_4\}$& $\{\Psi_4\}$&$\emptyset$&$\{\Psi_5,\Psi_4\}$&exceptional I\\  
Case x.&
$\{\Psi_5,\Psi_4\}$& $\{\Psi_4\}$&$\emptyset$&$\{\Psi_4\}$&exceptional I\\  
Case x\hspace{-.05em}i.&
$\{\Psi_3,\Psi_6\}$& $\{\Psi_3\}$&$\emptyset$&$\emptyset$&exceptional I\\ 
Case x\hspace{-.05em}i\hspace{-.05em}i.&
$\{\Psi_3,\Psi_4\}$& $\{\Psi_3,\Psi_4\}$&$\emptyset$&$\{\Psi_4\}$&exceptional I\\ 
\hline
\end{tabular}
\caption{Bases of the resulting eigenspaces (Proposition~\ref{12.12.19.6.32}).}
\label{13.3.6.0.28}
\end{center}
\end{table}
\end{proposition}

\begin{proposition}\label{12.12.19.6.33}
Suppose that $\beta\ge 4$ in Assumption~\ref{12.11.8.1.9},
and that $P$, $m_0$ and $q_0$ are not invertible in 
$\mathcal B(\mathcal K)$, $\mathcal B(Q\mathcal K)$ 
and $\mathcal B(S\mathcal K)$, respectively. Then
\begin{align}
r_0=-Tz^*zT,
\label{13.3.8.16.2}
\end{align}
and $r_0$ is always invertible in $\mathcal B(T\mathcal K)$.
Bases of the eigenspaces are given 
as in Table~\ref{13.3.6.0.29} according to the Cases
defined in Table~\ref{13.3.6.0.29a},
where the entries with parentheses are automatically determined by the assumption and 
those without.
\begin{table}[h]
\begin{center}
\begin{tabular}{|l|ccccc|}
\hline
& $\Phi_1$ & $\Phi_2$ & $\Phi_3$ & $\Phi_4$ &  $\Delta$ \\ 
\hline
Case i.& $= 0$& $= 0$&($=0$)&$= 0$&($=0$)\\ 
Case i\hspace{-.05em}i.& $\neq 0$& $= 0$&$= 0$&$= 0$&$= 0$\\ 
Case i\hspace{-.05em}i\hspace{-.05em}i.& $\neq 0$& $= 0$&$=0$&$= 0$&$\neq 0$\\ 
Case i\hspace{-.05em}v.& $= 0$& $\neq 0$&($=0$)&$= 0$&($=0$)\\ 
Case v.& $\neq 0$& $\neq 0$&$= 0$&$= 0$&$= 0$\\ 
Case v\hspace{-.05em}i.& $\neq 0$& $\neq 0$&$=0$&$= 0$&$\neq 0$\\ 
Case v\hspace{-.05em}i\hspace{-.05em}i.& $=0$& ($\neq 0$)&($=0$)&$\neq 0$&($=0$)\\ 
Case v\hspace{-.05em}i\hspace{-.05em}i\hspace{-.05em}i.& ($\neq 0$)& $= 0$&$\neq 0$&$= 0$&arbitrary\\ 
Case i\hspace{-.05em}x.& $\neq 0$& ($\neq 0$)&$=0$&($\neq 0$)&$=0$\\ 
Case x.& $\neq 0$& ($\neq 0$)&$=0$&$\neq 0$&$\neq 0$\\ 
Case x\hspace{-.05em}i.& ($\neq 0$)& $\neq 0$&$\neq 0$&$= 0$&arbitrary\\ 
Case x\hspace{-.05em}i\hspace{-.05em}i.& ($\neq 0)$& ($\neq 0$)&$\neq 0$&$\neq 0$&arbitrary\\ 
\hline
\end{tabular}
\caption{Definitions of Cases (Proposition~\ref{12.12.19.6.33}).}
\label{13.3.6.0.29a}
\end{center}
\end{table}
\begin{table}[h]
\begin{center}
\begin{tabular}{|l|ccccc|}
\hline
\rule{0cm}{0.5cm}& $\widetilde{\mathcal E}/E$ & $ \mathcal E/E$
 &  $ E$ & $\widetilde{\mathcal E}_{\mathrm{qs}}/E$ & type\\ 
\hline
Case i.& $\{\Psi_1^0,\Psi_2^0\}$& $\{\Psi_1^0\}$&$0<\#<\infty$&$\emptyset$& exceptional I\hspace{-.1em}I\hspace{-.1em}I \\ 
Case i\hspace{-.05em}i.& $\{\Psi_5,\Psi_2^0\}$& $\emptyset$&$0<\#<\infty$&$\{\Psi_5\}$&exceptional I\hspace{-.1em}I\\ 
Case i\hspace{-.05em}i\hspace{-.05em}i.& $\{\Psi_5,\Psi_2^0\}$& $\emptyset$&$0<\#<\infty$&$\emptyset$&exceptional I\hspace{-.1em}I\\ 
Case i\hspace{-.05em}v.& $\{\Psi_1^0,\Psi_6\}$& $\{\Psi_1^0\}$&$0<\#<\infty$&$\emptyset$&exceptional I\hspace{-.1em}I\hspace{-.1em}I\\ 
Case v.& $\{\Psi_5,\Psi_6\}$& $\emptyset$&$0<\#<\infty$&$\{\Psi_5\}$&exceptional I\hspace{-.1em}I\\ 
Case v\hspace{-.05em}i.& $\{\Psi_5,\Psi_6\}$& $\emptyset$&$0<\#<\infty$&$\emptyset$&exceptional I\hspace{-.1em}I\\ 
Case v\hspace{-.05em}i\hspace{-.05em}i.& $\{\Psi_1^0,\Psi_4\}$& $\{\Psi_1^0,\Psi_4\}$&$0<\#<\infty$&$\{\Psi_4\}$&exceptional I\hspace{-.1em}I\hspace{-.1em}I\\ 
Case v\hspace{-.05em}i\hspace{-.05em}i\hspace{-.05em}i.& $\{\Psi_3,\Psi_2^0\}$& $\{\Psi_3\}$&$0<\#<\infty$&$\emptyset$&exceptional I\hspace{-.1em}I\hspace{-.1em}I\\  
Case i\hspace{-.05em}x.& $\{\Psi_5,\Psi_4\}$& $\{\Psi_4\}$&$0<\#<\infty$&$\{\Psi_5,\Psi_4\}$&exceptional I\hspace{-.1em}I\hspace{-.1em}I\\  
Case x.& $\{\Psi_5,\Psi_4\}$& $\{\Psi_4\}$&$0<\#<\infty$&$\{\Psi_4\}$&exceptional I\hspace{-.1em}I\hspace{-.1em}I\\  
Case x\hspace{-.05em}i.& $\{\Psi_3,\Psi_6\}$& $\{\Psi_3\}$&$0<\#<\infty$&$\emptyset$&exceptional I\hspace{-.1em}I\hspace{-.1em}I\\ 
Case x\hspace{-.05em}i\hspace{-.05em}i.& $\{\Psi_3,\Psi_4\}$& $\{\Psi_3,\Psi_4\}$&$0<\#<\infty$&$\{\Psi_4\}$&exceptional I\hspace{-.1em}I\hspace{-.1em}I\\ 
\hline
\end{tabular}
\caption{Bases of the resulting eigenspaces (Proposition~\ref{12.12.19.6.33}).}
\label{13.3.6.0.29}
\end{center}
\end{table}
\end{proposition}

The finiteness of the inversion iteration
discussed in Section~\ref{12.11.24.5.48}
also follows directly from Proposition~\ref{12.12.19.6.33}.
We state it as a corollary here.
The proof is straightforward, if we combine 
Proposition~\ref{12.12.19.6.33} with (\ref{12.12.19.8.42}) and (\ref{12.12.19.6.58}).
\begin{corollary}\label{12.12.28.10.55}
Under the assumption of Proposition~\ref{12.12.19.6.33}
the identities $r_0^\dagger r_0=r_0 r_0^\dagger=T$ hold, i.e., 
$r_0$ is invertible in $\mathcal B(T\mathcal K)$,
and hence the pseudoinverse $r(\kappa)^{\dagger}$ 
can be computed by the Neumann series.
\end{corollary}
In the remainder  of this subsection we prove 
Propositions~\ref{12.12.19.6.30}--\ref{12.12.19.6.33}.

\begin{proof}[Proof of Proposition~\ref{12.12.19.6.30}]
The first part of the assertion is obvious.

Assume $P$ is invertible in $\mathcal B(\mathcal K)$,
Then by Corollary~\ref{13.1.16.2.51} it immediately follows that 
\begin{align*}
E=\{0\}.
\end{align*}
In order to identify the remaining 
eigenspaces $\widetilde{\mathcal E}$, $\mathcal E$ and $\widetilde{\mathcal E}_{\mathrm{qs}}$
we separate the cases
and formally apply the block decomposition argument for $M_0$.
Here, in fact, the decomposition argument is 
not necessary, but this goes along with the following proofs.
As in Remark~\ref{13.3.20.7.9},
we can write the corresponding block matrix in the following manner.
Then noting Corollary~\ref{13.1.16.2.51} and Lemma~\ref{13.3.9.3.47},
we can complete Table~\ref{13.3.6.0.26} with ease.

\smallskip
\noindent
\textit{Case i.} 
\begin{align*}
\mathcal K=\{0\},\quad
M_0=0.
\end{align*}
\smallskip
\noindent
\textit{Cases i\hspace{-.05em}i and i\hspace{-.05em}i\hspace{-.05em}i.} 
\begin{align*}
\mathcal K=\mathbb C \Phi_1,\quad
M_0=
\left(\begin{array}{c}
?
\end{array}\right).
\end{align*}
\end{proof}
\begin{proof}[Proof of Proposition~\ref{12.12.19.6.31}]
The first identity of (\ref{12.12.19.7.0}) is due to the definition $m_0=QM_0Q$.
The latter is also already known by the isomorphism (\ref{12.12.19.6.57})
and the injectivity of $z$ on $\mathop{\mathrm{Ker}}\widetilde QM_0$ 
due to Proposition~\ref{12.12.27.16.6}.

Let us assume that the operator $m_0$ is invertible in $\mathcal B(Q\mathcal K)$.
By Corollary~\ref{13.1.16.2.51} it follows 
\begin{align*}
E=\{0\}.
\end{align*}
We investigate the remaining eigenspaces employing the block decomposition argument
as in the previous proof.
We write
\begin{align*}
\mathcal K
&=\mathbb C \Phi_1\oplus S^\perp Q\mathcal K,\quad
M_0=
\left(\begin{array}{cc}
?&?\\
?&*
\end{array}\right),
\end{align*}
and proceed to Cases i--v\hspace{-.05em}i.
We note that for Cases i--i\hspace{-.05em}i\hspace{-.05em}i,
compared to the previous proof,
there appear extra column and row concerning $S^\perp Q\mathcal K$-factor,
but these do not essentially affect the argument
and we obtain the same list of eigenfunctions as before.
Hence, we omit Cases~i--i\hspace{-.05em}i\hspace{-.05em}i,
and consider only Cases~i\hspace{-.05em}v--v\hspace{-.05em}i
without distinguishing whether $S^\perp Q\mathcal K$ is trivial
by the same reason.
Once we write down the matrix expressions for $M_0$,
then the assertion is straightforward
by Corollary~\ref{13.1.16.2.51}, Lemma~\ref{13.3.9.3.47} and (\ref{13.3.9.14.47b}).

\smallskip
\noindent
\textit{Case i\hspace{-.05em}v}.
\begin{align*}
\mathcal K
&=S^\perp Q\mathcal K,\quad
M_0=
\left(\begin{array}{c}
*
\end{array}\right).
\end{align*}

\smallskip
\noindent
\textit{Cases v and v\hspace{-.05em}i}.
\begin{align*}
\mathcal K
&=\mathbb C \Phi_1\oplus S^\perp Q\mathcal K,\quad
M_0=
\left(\begin{array}{ccc}
?&?\\
?&*
\end{array}\right).
\end{align*}
\end{proof}

\begin{proof}[Proof of Proposition~\ref{12.12.19.6.32}]
In order to prove (\ref{12.12.19.8.41})
we use the expressions of $q_0$ and $m_1$ 
in Lemma~\ref{12.12.19.4.4}.
Noting that $QM_0S$ and $SM_0Q$ are trivial by the definition of $S$,
we can write
\begin{align*}
q_0=SM_1S-\gamma SM_0PM_0S.
\end{align*}
We substitute (\ref{11.2.1.1.40}) for $M_0$ and $M_1$ above,
and use the kernel expressions (\ref{12.12.19.8.26})
for $G_0^0$ and $G_1^0$, respectively.
Then by $Sv^*\mathbf 1=Qv^*\mathbf 1=0$ we obtain 
\begin{align*}
q_0
=-\tfrac12 \langle Sv^*\mathbf n,\cdot\rangle Sv^*\mathbf n
-\tfrac{\gamma^2}2 \langle SM_0v^*\mathbf 1,\cdot\rangle SM_0v^*\mathbf 1.
\end{align*}
Thus (\ref{12.12.19.8.41}) is verified.

Next, we prove (\ref{12.12.19.8.42}).
First note that 
\begin{align*}
S\mathcal K=\mathop{\mathrm{Ker}}m_0
=\mathop{\mathrm{Ker}} P\cap \mathop{\mathrm{Ker}}  QM_0
\supset \mathop{\mathrm{Ker}} \widetilde P
\cap \mathop{\mathrm{Ker}}  M_0,
\end{align*}
and hence it suffices to consider only $\Phi \in S\mathcal K$.
By (\ref{12.12.19.8.41}) it follows that, 
for $\Phi \in S\mathcal K$,
the identity $q_0(\Phi)=0$ holds if and only if
\begin{align}
\langle SM_0v^*\mathbf 1,\Phi\rangle =\langle Sv^*\mathbf n,\Phi\rangle=0.
\label{12.12.19.8.57}
\end{align}
Note that this equivalence holds 
whether $Sv^*\mathbf n$ and $SM_0v^*\mathbf 1$ are 
linearly independent or not.
The condition (\ref{12.12.19.8.57}) for $\Phi \in S\mathcal K$ 
is equivalent to 
\begin{align*}
\Phi\in \mathop{\mathrm{Ker}} \widetilde P\cap \mathop{\mathrm{Ker}}  M_0,
\end{align*}
which implies the first identity of (\ref{12.12.19.8.42}).
The second isomorphism is due to (\ref{12.12.19.6.58}).

Now we assume $q_0$ is invertible in $\mathcal B(S\mathcal K)$.
By (\ref{12.12.19.8.42}) it follows 
\begin{align*}
E=\{0\}.
\end{align*}
The remaining eigenspaces are studied by the block decomposition argument,
again.
Note that by the above argument we have a further decomposition
\begin{align*}
S\mathcal K
&=\mathbb C \Phi_3
\oplus \mathbb C \Phi_4,
\end{align*}
and also that 
\begin{align}
M_0\Phi_4=0.\label{13.3.20.14.36}
\end{align}
By this and Lemma~\ref{13.3.9.3.47}
we can calculate the matrix expressions for $M_0$, and
complete Table~\ref{13.3.6.0.28} as in the previous proofs.
Note that we do not have to distinguish whether
$S^\perp Q\mathcal K$ is trivial or not, 
as before.

\smallskip
\noindent
\textit{Case v\hspace{-.05em}i\hspace{-.05em}i}.
\begin{align*}
\mathcal K
=
S^\perp Q\mathcal K\oplus
\mathbb C \Phi_4,\quad
M_0=
\left(\begin{array}{cc}
*&0\\
0&0
\end{array}\right).
\end{align*}

\smallskip
\noindent
\textit{Cases v\hspace{-.05em}i\hspace{-.05em}i\hspace{-.05em}i and x\hspace{-.05em}i}.
\begin{align*}
\mathcal K
&=\mathbb C \Phi_1\oplus S^\perp Q\mathcal K\oplus\mathbb C \Phi_3,\quad
M_0=
\left(\begin{array}{ccc}
?&?&*\\
?&*&0\\
*&0&0
\end{array}\right).
\end{align*}

\smallskip
\noindent
\textit{Case i\hspace{-.05em}x and x}.
\begin{align*}
\mathcal K
&=P\mathcal K\oplus S^\perp Q\mathcal K\oplus\mathbb C \Phi_4,\quad
M_0=
\left(\begin{array}{ccc}
?&?&0\\
?&*&0\\
0&0&0
\end{array}\right).
\end{align*}

\smallskip
\noindent
\textit{Case x\hspace{-.05em}i\hspace{-.05em}i}.
\begin{align*}
\mathcal K
&=\mathbb C \Phi_1\oplus S^\perp Q\mathcal K\oplus
\mathbb C \Phi_3
\oplus \mathbb C \Phi_4,\quad
M_0=
\left(\begin{array}{cccc}
?&?&*&0\\
?&*&0&0\\
*&0&0&0\\
0&0&0&0
\end{array}\right).
\end{align*}
\end{proof}

Before the proof of Proposition~\ref{12.12.19.6.33}
we show the following lemma:
\begin{lemma}\label{12.12.27.15.3}
Let $x_j\in \mathcal L^4$ for $j=1,2$
satisfy
\begin{align}
\langle \mathbf 1, x_j\rangle =\langle \mathbf n,x_j\rangle=0.
\label{12.12.27.7.13}
\end{align}
Then $G_0^0x_j\in \mathcal L^2$ for $j=1,2$, and
\begin{align}
\langle x_1,G_2^0x_2\rangle =-\langle G_0^0x_1,G^0_0x_2\rangle.
\label{13.3.8.14.53}
\end{align}
\end{lemma}
\begin{proof}
The assertion $G_0^0x_j\in\mathcal L^2$ for $j=1,2$
follows from Lemma~\ref{12.11.24.18.24},
so that the right-hand side of (\ref{13.3.8.14.53}) makes sense.
By the kernel expressions (\ref{12.12.19.8.26}) 
and the assumption (\ref{12.12.27.7.13}) we have
\begin{align*}
\langle x_1,G_{-1}^0x_2\rangle=\langle x_1,G_{1}^0x_2\rangle=0,
\end{align*}
and hence by Proposition~\ref{12.11.9.6.23} with $N=2$,
where the large weight $4$ is required,
\begin{align}
\langle x_1,G_2^0x_2\rangle 
=\lim_{\kappa\to +0}\tfrac1{\kappa^2}
\langle x_1,(R_0(\kappa)-G_0^0)x_2\rangle.
\label{12.12.27.11.35}
\end{align}
We now use the Fourier transform to compute the right-hand side.
We already know the Fourier transform of $R(\kappa)$ by (\ref{12.12.27.12.55}),
and thus consider that of $G_0^0$. 
We claim
\begin{align}
\langle x_1,G_0^0x_2\rangle
=\!\!\int_{\mathbb T}
\frac{\overline{\hat x_1(\theta)}\hat x_2(\theta)}{4\sin^2(\theta/2)}
\,d\theta.
\label{12.12.27.11.17}
\end{align}
Note that the right-hand side is actually a convergent integral due to
\begin{align}
\hat{x}_j\in C^4(\mathbb T),\quad  \hat{x}_j(0)=0,\quad
\hat{x}_j'(0)=0.\label{12.12.27.11.30}
\end{align}
We consider the distribution $f\in\mathcal D'(\mathbb T)$ defined by
\begin{align*}
f(\theta)=-(8\pi)^{-1/2}
\tfrac{d}{d\theta}\mathrm{p.v.}\cot\tfrac{\theta}2,
\end{align*}
and then, using Cauchy's integral formula, we can verify 
\begin{align*}
\mathcal F^{-1}(f)[n]=-\tfrac12|n|=G_0^0(n).
\end{align*}
Hence, for rapidly decreasing the sequences $y_j$, $j=1,2$, 
or equivalently, $\hat y_j\in C^\infty(\mathbb T)$, 
we obtain 
\begin{align}
\langle y_1,G_0^0y_2\rangle
=\sqrt{2\pi }\langle \hat y_1,\hat y_2\cdot f\rangle
=\tfrac12\mathrm{p.v.}\!\!\int_{\mathbb T}
\bigl(\overline{\hat y_1(\theta)}\hat y_2(\theta)\bigr)'\cot\tfrac{\theta}2
\,d\theta.
\label{12.12.27.11.0}
\end{align}
By the density argument the identity (\ref{12.12.27.11.0})
extends to $y_j\in \mathcal L$, $j=1,2$.
If we set $y_j=x_j$ of the lemma,
then we obtain the claim (\ref{12.12.27.11.17}) by integration by parts.
We note that we can drop `$\mathrm{p.v.}$' using (\ref{12.12.27.11.30}).

We note that from the claim (\ref{12.12.27.11.17})
and the density argument
a slightly general identity follows :
\begin{align}
\mathcal F(G_0^0x_j)(\theta)
=\frac{\hat x_j(\theta)}{4\sin^2(\theta/2)}.\label{13.3.8.15.24}
\end{align}
By (\ref{12.12.27.11.35}), (\ref{12.12.27.12.55}) and (\ref{12.12.27.11.17}) we can write
\begin{align*}
\langle x_1,G_2^0x_2\rangle
=-\lim_{\kappa\to +0}
\int_{\mathbb T}
\frac{\overline{\hat x_1(\theta)}\hat x_2(\theta)}{
4\sin^2(\theta/2)(4\sin^2(\theta/2)+\kappa^2)}\,
d\theta.
\end{align*}
By (\ref{12.12.27.11.30}) we can apply the dominated convergence theorem,
and then by (\ref{13.3.8.15.24}) we obtain the assertion.
\end{proof}

\begin{proof}[Proof of Proposition~\ref{12.12.19.6.33}]
We first note
\begin{align}
TM_0=0,\quad TM_1Q=0.\label{12.12.27.14.14}
\end{align}
In fact, we can rewrite using 
$1_{\mathcal K}=P+Q$, 
(\ref{11.2.1.1.40}) and (\ref{12.12.19.8.26})
\begin{align*}
TM_0=\tfrac{\gamma}2\langle v^*\mathbf 1,\cdot\rangle TM_0v^*\mathbf 1
+TM_0Q,\quad
TM_1Q=-\tfrac12 \langle Qv^*\mathbf n,\cdot\rangle Tv^*\mathbf n,
\end{align*}
and then (\ref{12.12.27.14.14}) immediately follow.
By Lemma~\ref{12.12.19.4.4}, 
(\ref{12.12.27.14.14}) and its adjoint
we can compute
\begin{align*}
r_0=Tv^*G_2^0vT.
\end{align*}
Since the image $vT(\mathcal K)\subset \mathcal L^\beta$ is orthogonal to 
both $\mathbf 1$ and $\mathbf n$ with respect to the $\ell^2$-inner product,
we can apply Lemma~\ref{12.12.27.15.3} and then obtain
\begin{align*}
r_0=-Tv^*G_0^0G_0^0vT=-Tz^*zT.
\end{align*}
Here we replaced $-G_0^0v$ by $z$ using the first identity of 
(\ref{12.12.27.14.14}). Thus we obtain (\ref{13.3.8.16.2}).

Now we investigate the eigenspaces.
We already know 
\begin{align*}
E\neq \{0\},\quad 0< d<\infty.
\end{align*}
As in the previous proofs, we apply the block decomposition argument.
We decompose
\begin{align*}
\mathcal K
&=\mathbb C \Phi_1\oplus S^\perp Q\mathcal K\oplus
\mathbb C \Phi_3
\oplus \mathbb C \Phi_4\oplus T\mathcal K,
\end{align*}
and
\begin{align}
M_0=
\left(\begin{array}{ccccc}
?&?&*&0&0\\
?&*&0&0&0\\
*&0&0&0&0\\
0&0&0&0&0\\
0&0&0&0&0
\end{array}\right)
\label{13.3.8.16.21}
\end{align}
accordingly.
The representation (\ref{13.3.8.16.21}) can be verified, e.g., by 
Lemma~\ref{13.3.9.3.47}, (\ref{13.3.20.14.36}) and (\ref{12.12.27.14.14}).
Fortunately,
the space $T\mathcal K$ does not have any effects under the action of $M_0$, 
and hence we can reduce Cases~i--x\hspace{-.05em}i\hspace{-.05em}i
to those of the previous propositions.
Thus we are done.
\end{proof}

\section{Expansion of $R(\kappa)$}\label{12.12.17.20.56}

\subsection{The case: $P$ is invertible}

\begin{theorem}\label{12.12.1.12.38}
Suppose that $\beta\ge 1$ in Assumption~\ref{12.11.8.1.9},
and that $P$ defined by (\ref{12.11.24.22.25}) is invertible in $\mathcal B(\mathcal K)$.
Let $N\in [-1,\beta-2]$ be any integer.
Then as $\kappa\to 0$ with $\real\kappa>0$,
the resolvent $R(\kappa)$ has the expansion in $\mathcal B^{N+2}$:
\begin{align}
R(\kappa)=\sum_{j=-1}^N\kappa^j G_j+\mathcal O(\kappa^{N+1}),\quad
G_j\in \mathcal B^{j+1},
\label{13.3.11.4.27}
\end{align}
and the coefficients $G_j$ can be computed explicitly.
\end{theorem}
\begin{proof}
Since the leading operator of the expansion of $M(\kappa)$ is invertible,
we can expand $M(\kappa)^{-1}$ employing the Neumann series.
Using Proposition~\ref{12.11.27.11.18} for any integer $N\in [-1,\beta-2]$, 
we have the expansion in $\mathcal B(\mathcal K)$:
\begin{align}
M(\kappa)^{-1}=\sum_{j=1}^{N+2}\kappa^jA_j+\mathcal O(\kappa^{N+3}),\quad
A_j\in\mathcal B(\mathcal K).
\label{12.12.1.7.6}
\end{align}
The coefficients $A_j$ can be written explicitly in terms of $M_j$:
\begin{align}
A_1&=\gamma ,\quad 
A_{j}=-\sum_{k=1}^{j-1}(-\gamma)^{k+1}
\sum_{\genfrac{}{}{0pt}{}{j_1\ge 0,\dots,j_k\ge 0}{ j_1+\cdots+j_k=j-k-1}}\prod_{l=1}^kM_{j_l}
\quad\mbox{for } j\ge 2.\label{12.12.1.7.18}
\end{align}
Let us substitute the expansions (\ref{12.12.1.7.6}) and (\ref{12.12.1.7.7}) 
into (\ref{12.11.8.4.1}).
Then for the same $N$ as above we have the expansion
in $\mathcal B^{N+2}$:
\begin{align*}
R(\kappa)=\sum_{j=-1}^N\kappa^jG^0_j
-\sum_{j=-1}^N\kappa^j\sum_{\genfrac{}{}{0pt}{}{j_1\ge -1,j_2\ge 1,j_3\ge -1}{j_1+j_2+j_3=j}}G_{j_1}^0vA_{j_2}v^*G_{j_3}^0
+\mathcal O(\kappa^{N+1}).
\end{align*}
Hence, we obtain the expansion of $R(\kappa)$ by putting for $j\ge -1$
\begin{align}
G_j=G_j^0-\sum_{\genfrac{}{}{0pt}{}{j_1\ge -1,j_2\ge 1,j_3\ge -1}{j_1+j_2+j_3=j}}
G_{j_1}^0vA_{j_2}v^*G_{j_3}^0\in \mathcal B^{j+1}.\label{12.12.28.17.4}
\end{align}
Hence, we are done.
\end{proof}

We next discuss the relation between the coefficients and the eigenspaces.
By (\ref{12.12.28.17.4}) and (\ref{12.12.1.7.18})
\begin{align}
\begin{split}
G_{-1}&=G_{-1}^0-G_{-1}^0vA_1v^*G_{-1}^0\\
&=G_{-1}^0-\gamma G_{-1}^0vv^*G_{-1}^0,
\end{split}\label{12.12.28.18.4}\\
\begin{split}
G_{0}&=G_0^0
-G_{-1}^0vA_2v^*G_{-1}^0
-G_{-1}^0vA_1v^*G_{0}^0
-G_{0}^0vA_1v^*G_{-1}^0\\
&=G_0^0+\gamma^2 G_{-1}^0vM_0v^*G_{-1}^0-\gamma G_{-1}^0vv^*G_{0}^0-\gamma G_{0}^0vv^*G_{-1}^0.
\end{split}\label{12.12.28.18.5}
\end{align}

\begin{theorem}\label{13.3.11.5.15}
Under the assumption of Theorem~\ref{12.12.1.12.38}
the coefficient $G_{-1}$ is of the form
(\ref{13.3.11.2.46}) (up to the factor $i$ due to a different convention).
\end{theorem}
\begin{remark}
If $V\neq0$ is local, then we always have $G_{-1}=0$, see 
Section~\ref{local-P-inv}.
\end{remark}
\begin{proof}
We already have the list of the eigenfunctions
in Table~\ref{13.3.6.0.26} for all Cases~i--i\hspace{-.05em}i\hspace{-.05em}i in Table~\ref{13.3.6.0.26a}.
Then since $G_{-1}=2^{-1}\langle \Psi_1^0,\cdot\rangle \Psi_1^0$ for Case~i and $G_{-1}=0$ for Cases~i\hspace{-.05em}i and i\hspace{-.05em}i\hspace{-.05em}i,
the assertion follows immediately.
\end{proof}

Let us consider $G_0$.
It is obvious that we have 
\begin{align*}
G_{0}
\equiv 
G_0^0
+\langle \Psi_5,\cdot \rangle  \Psi_1^0
+\langle \Psi_1^0,\cdot \rangle  \Psi_5
\mod \mathcal B^0.
\end{align*}
By a direct computation, e.g., with $v^*G_{-1}^0v=\gamma^\dagger$, $P=1_{\mathcal K}$ and $\Phi_1=\Phi_5$,
we obtain
\begin{align}
(1_{\mathcal L^*}+G_0^0V)G_0
&=G_0^0+\gamma^2 G_{-1}^0vM_0v^*G_{-1}^0-\gamma G_{-1}^0vv^*G_{0}^0-\gamma G_{0}^0vv^*G_{-1}^0\notag\\
&\quad+G_0^0VG_0^0+\gamma G_0^0vUPM_0v^*G_{-1}^0-G_0^0vUPv^*G_{0}^0\notag\\
&\quad-\gamma G_0^0VG_{0}^0vv^*G_{-1}^0\notag\\
&=G_0^0
+ \langle \Psi_5,\cdot \rangle \Psi_1^0.
\label{13.1.19.14.50}
\end{align}
By (\ref{13.1.19.14.50}) and Lemma~\ref{13.1.18.6.0} it immediately follows that 
\begin{align*}
HG_0=1_{\mathcal L}.
\end{align*}
A similar adjoint computation applies to $G_0H$ with Lemma~\ref{13.1.18.6.0}, 
and verifies $G_0H=1_{\mathcal L}$.
If $\widetilde{\mathcal E}_{\mathrm{qs}}=\{0\}$, then 
$1_{\mathcal L^*}+G_0^0V$ is invertible
and it follows by (\ref{13.1.19.14.50}) and (\ref{13.3.9.14.47}) that
\begin{align*}
\begin{split}
G_0
=(1_{\mathcal L^*}+G_0^0V)^{-1}G_0^0
+\Delta^{\dagger}
\langle \Psi_5,\cdot\rangle \Psi_5.
\end{split}
\end{align*}
Even if $\widetilde{\mathcal E}_{\mathrm{qs}}\neq \{0\}$,
we can still modify the argument.
Set 
\begin{align*}
\begin{split}
\pi_1&=1_{\mathcal L}-\langle \Psi_1^0,\cdot\rangle V\Psi_5,
\quad
\pi_2=1_{\mathcal L}-\langle \Psi_5,\cdot\rangle V\Psi_1^0.
\end{split}
\end{align*}
Since $\langle \Psi_1^0,V\Psi_5\rangle=\langle \Psi_5,V\Psi_1^0\rangle = 1\neq 0$,
the operators $\pi_1$ and $\pi_2$ are in fact projections on the Banach space $\mathcal L$.
Then noting 
$\mathop{\mathrm{Ran}}(1_{\mathcal L}+G_0^0V)=\pi_1^*\mathcal L$
and $\mathop{\mathrm{Ker}}(1_{\mathcal L^*}+G_0^0V)=\mathbb C\Psi_5$
by (\ref{13.3.6.9.36}) and Table~\ref{13.3.6.0.26},
we can write (\ref{13.1.19.14.50}) as
\begin{align}
(1_{\mathcal L^*}+G_0^0V)\pi_2^*G_0=\pi_1^*G_0^0.
\label{13.8.2.17.00}
\end{align}
Now recall the Riesz--Schauder theorem.
Since $\langle \Psi_5,V\Psi_1^0\rangle\neq 0$,
we can invert $\pi_1^*(1_{\mathcal L^*}+G_0^0V)\pi_2^*$
between the projected spaces. 
We may write (\ref{13.8.2.17.00}) as
\begin{align*}
\pi_2^*G_0
=[\pi_1^*(1_{\mathcal L^*}+G_0^0V)\pi_2^*]^{-1}\pi_1^*G_0^0.
\end{align*}
On the other hand, we can compute from (\ref{12.12.28.18.5}) with $M_0=0$ and $U=\pm 1_{\mathcal K}$ that 
\begin{align*}
(1-\pi_2^*)G_0=\langle \Psi_1^0,\cdot\rangle\Psi_5.
\end{align*}
Hence we obtain the expression
\begin{align*}
G_0
=[\pi_1^*(1_{\mathcal L^*}+G_0^0V)\pi_2^*]^{-1}\pi_1^*G_0^0+\langle \Psi_1^0,\cdot\rangle\Psi_5
.
\end{align*}
We gather the above results as a theorem.
The adjoint statements are obtained similarly.

\begin{theorem}\label{13.3.11.4.29}
Under the assumption of Theorem~\ref{12.12.1.12.38}
the coefficient $G_0$ satisfies
\begin{align*}
G_{0}
\equiv 
G_0^0
+\langle \Psi_5,\cdot \rangle  \Psi_1^0
+\langle \Psi_1^0,\cdot \rangle  \Psi_5
\mod \mathcal B^0
\end{align*}
and 
\begin{align*}
HG_0=G_0H=1_{\mathcal L}.
\end{align*}
Moreover, with the above notation, 
if $\widetilde{\mathcal E}_{\mathrm{qs}}=\{0\}$, then
\begin{align}
\begin{split}
G_0
&=(1_{\mathcal L^*}+G_0^0V)^{-1}G_0^0
+\Delta^{\dagger}
\langle \Psi_5,\cdot\rangle \Psi_5\\
&=G_0^0(1_{\mathcal L}+VG_0^0)^{-1}
+\Delta^{\dagger}
\langle \Psi_5,\cdot\rangle \Psi_5,
\end{split}\nonumber
\intertext{and, if $\widetilde{\mathcal E}_{\mathrm{qs}}\neq \{0\}$, then }
\begin{split}
G_0
&=[\pi_1^*(1_{\mathcal L^*}+G_0^0V)\pi_2^*]^{-1}\pi_1^*G_0^0
+\langle \Psi_1^0,\cdot\rangle\Psi_5\\
&=G_0^0\pi_1[\pi_2(1_{\mathcal L}+VG_0^0)\pi_1]^{-1}
+
\langle \Psi_5,\cdot\rangle\Psi_1^0.
\end{split}
\label{13.1.19.14.52}
\end{align}
\end{theorem}
\begin{remarks}\label{13.3.11.5.22}
\begin{enumerate}
\item
The coefficient $G_0$ appears in the expansion (\ref{13.3.11.4.27})
for $N\ge 0$, so that we seemingly need $\beta\ge 2$,
but do not in fact, because that is required to dominate the error term. 
If we take (\ref{12.12.28.18.5}) as the definition of $G_0$,
we only need $\beta\ge 1$. Theorem~\ref{13.3.11.4.29} reads in this sense.

\item
As in Lemma~\ref{13.1.18.6.0}, the composition $G_0H$ makes sense 
also on the extended space $\mathbb C\mathbf n\oplus 
\mathbb C|\mathbf n|\oplus
\mathbb C\mathbf 1\oplus
\mathbb C\mbox{\boldmath$\sigma $}
\oplus\mathcal L$,
and we can compute the action on $\mathbb C\mathbf n\oplus 
\mathbb C|\mathbf n|\oplus
\mathbb C\mathbf 1\oplus
\mathbb C\mbox{\boldmath$\sigma $}$
by a direct computation employing (\ref{12.12.28.18.5}).
We here write down only the consequences, omitting the detail:
\begin{align*}
G_{0}H\Psi_1^0
&=\|\Phi_1\|^{\dagger 2}\|\Phi_1\|^2\Psi_1^0,&
G_{0}H\Psi_2^0&=0,\\
G_{0}H\mbox{\boldmath$\sigma$}&=\mbox{\boldmath$\sigma$},&
G_{0}H|\mathbf n|
&=|\mathbf n|-2\Psi_5.
\end{align*}
\end{enumerate}
\end{remarks}

\subsection{The case: $P$ is not invertible, and $m_0$ is invertible}
\begin{theorem}\label{12.12.1.12.39}
Suppose that $\beta\ge 2$ in Assumption~\ref{12.11.8.1.9},
that $P$ defined by (\ref{12.11.24.22.25}) is not invertible in $\mathcal B(\mathcal K)$,
and that $m_0$ defined by (\ref{12.11.27.8.11}) is invertible in $\mathcal B(Q\mathcal K)$.
Let $N\in [0,\beta-2]$ be any integer.
Then as $\kappa\to 0$ with $\real\kappa>0$,
the resolvent $R(\kappa)$ has the expansion in $\mathcal B^{N+2}$:
\begin{align}
R(\kappa)=\sum_{j=-1}^N\kappa^j G_j+\mathcal O(\kappa^{N+1}),\quad
G_j\in \mathcal B^{j+1},\label{12.12.1.3.7}
\end{align}
and the coefficients $G_j$ can be computed explicitly.
\end{theorem}
\begin{proof}
By Proposition~\ref{12.11.13.19.5} and (\ref{12.11.11.12.35}) we can write
\begin{align}
\begin{split}
R(\kappa)={}&R_0(\kappa)-\kappa R_0(\kappa)v(Q+\kappa M(\kappa))^{-1}v^*R_0(\kappa)\\
&-R_0(\kappa)v(Q+\kappa M(\kappa))^{-1}m(\kappa)^\dagger(Q+\kappa M(\kappa))^{-1}v^*R_0(\kappa).
\end{split}\label{12.11.30.5.11}
\end{align}
We compute the expansions of $m(\kappa)^\dagger$ and $(Q+\kappa M(\kappa))^{-1}$ 
employing the Neumann series.
Let $N\in [0,\beta-2]$ be any integer as in the assertion.
Then we have the expansions in $\mathcal B(\mathcal K)$
and $\mathcal B(Q\mathcal K)$:
\begin{align}
(Q+\kappa M(\kappa))^{-1}&
=\sum_{j=0}^{N+1}\kappa^j A_j+\mathcal O(\kappa^{N+2}),\quad
A_j\in\mathcal B(\mathcal K),\label{12.11.30.5.13}\\
m(\kappa)^\dagger
&=\sum_{j=0}^{N}\kappa^jB_j+\mathcal O(\kappa^{N+1}),\quad
B_j\in\mathcal B(Q\mathcal K),\label{12.11.30.5.12}
\end{align}
where
\begin{align}
A_0&=Q+\gamma P,&
A_{j}&=\sum_{k=1}^{j}
\sum_{\genfrac{}{}{0pt}{}{j_1\ge 0,\dots,j_k\ge 0}{j_1+\cdots+j_k=j-k}}A_0
\prod_{l=1}^k(-M_{j_l}A_0)\quad\mbox{for } j\ge 1,\label{12.11.30.7.20}\\
B_0&=m_0^\dagger, &
B_{j}&=\sum_{k=1}^{j}
\sum_{\genfrac{}{}{0pt}{}{j_1\ge 1,\dots,j_k\ge 1}{j_1+\cdots+j_k=j}}B_0
\prod_{l=1}^k(-m_{j_l}B_0)
\quad \mbox{for } j\ge 1,\label{12.11.30.7.19}
\end{align}
respectively.
We substitute (\ref{12.12.1.7.7}), (\ref{12.11.30.5.13}) and (\ref{12.11.30.5.12})
into (\ref{12.11.30.5.11}),
and then in the topology of $\mathcal B^{N+2}$
\begin{align}
\begin{split}
R(\kappa)={}&\sum_{j=-1}^N\kappa^jG_j^0
-
\sum_{j=-1}^{N}\kappa^j\sum_{\genfrac{}{}{0pt}{}{j_1\ge -1, j_2\ge 0,j_3\ge -1}{j_1+j_2+j_3=j-1}}
G_{j_1}^0vA_{j_2}v^*G_{j_3}^0
\\
&-
\sum_{j=-2}^{N}\kappa^j\sum_{\genfrac{}{}{0pt}{}{j_1\ge -1, j_2\ge 0,j_3\ge 0,j_4\ge 0,j_5\ge -1}{j_1+j_2+j_3+j_4+j_5=j}}
G_{j_1}^0vA_{j_2}B_{j_3}A_{j_4}v^*G_{j_5}^0
+\mathcal O(\kappa^{N+1}).
\end{split}\label{12.12.1.8.0}
\end{align}
Here we note that the error term is not $\mathcal O(\kappa^{N-1})$ or $\mathcal O(\kappa^{N})$,
because 
the possible error terms of order $\kappa^{N-1}$ or $\kappa^N$
contain the factors $G_{-1}v^*A_0Q$ or $QA_0vG_{-1}^0$,
both of which vanish:
\begin{align}
G_{-1}vA_0Q=0,\quad QA_0v^*G_{-1}^0=0.\label{12.11.30.6.34}
\end{align}
The same fact (\ref{12.11.30.6.34})
also guarantees that the coefficient of $\kappa^{-2}$ to the right of (\ref{12.12.1.8.0}) vanishes, 
and furthermore that, 
if we set for $j\ge -1$
\begin{align}
\begin{split}
G_j={}&G_j^0-\sum_{\genfrac{}{}{0pt}{}{j_1\ge -1, j_2\ge 0,j_3\ge -1}{j_1+j_2+j_3=j-1}}G_{j_1}^0vA_{j_2}v^*G_{j_3}^0
\\
&-\sum_{\genfrac{}{}{0pt}{}{j_1\ge -1, j_2\ge 0,j_3\ge 0,j_4\ge 0,j_5\ge -1}{j_1+j_2+j_3+j_4+j_5=j}}
G_{j_1}^0vA_{j_2}B_{j_3}A_{j_4}v^*G_{j_5}^0,
\end{split}\label{12.12.1.8.3}
\end{align}
then $G_j\in \mathcal B^{j+1}$.
Hence we obtain the desired expansion (\ref{12.12.1.3.7}).
\end{proof}

We next investigate the coefficients.
We can write by (\ref{12.12.1.8.3}), (\ref{12.11.30.6.34}),
(\ref{12.11.30.7.20}) and (\ref{12.11.30.7.19})
\begin{align}
\begin{split}
G_{-1}&=G_{-1}^0-G_{-1}^0vA_0v^*G_{-1}^0\\
&=G_{-1}^0-\gamma G_{-1}^0vv^*G_{-1}^0,
\end{split}
\label{13.1.19.16.22} \\
\begin{split}
G_{0}&=G_0^0
-G_{-1}^0vA_1v^*G_{-1}^0
-G_{-1}^0vA_0v^*G_{0}^0
-G_{0}^0vA_0v^*G_{-1}^0\\
&\phantom{={}}
-(G_{-1}^0vA_1+G_{0}^0vA_0)B_0(A_1v^*G_{-1}^0+A_0v^*G_{0}^0)\\
&=
G_0^0
+\gamma^2 G_{-1}^0vM_0v^*G_{-1}^0
-\gamma G_{-1}^0vv^*G_0^0
-\gamma G_0^0vv^*G_{-1}^0
\\
&\phantom{={}}
-(\gamma G_{-1}^0vM_0-G_0^0v)m_0^\dagger(\gamma M_0v^*G_{-1}^0-v^*G_{0}^0).
\end{split}
\label{13.1.19.16.23}
\end{align}
We note that parts of expressions 
(\ref{13.1.19.16.22}), (\ref{13.1.19.16.23})
coincide with those of (\ref{12.12.28.18.4}), (\ref{12.12.28.18.5}), respectively,
because parts of (\ref{12.12.1.8.3}) and (\ref{12.12.28.17.4})
are the same.
But we have to be careful since 
the definitions (\ref{12.11.30.7.20}) and (\ref{12.12.1.7.18}) of $A_j$ 
are different.

\begin{theorem}\label{13.3.11.5.16}
Under the assumption of Theorem~\ref{12.12.1.12.39}
the coefficient $G_{-1}$ is of the form (\ref{13.3.11.2.46})
(up to the factor $i$ due to a different convention).
\end{theorem}
\begin{proof}
With the bases listed in Table~\ref{13.3.6.0.27}
for Cases i--v\hspace{-.05em}i in Table~\ref{13.3.6.0.27a}
the assertion is obvious as that of Theorem~\ref{13.3.11.5.15}.
\end{proof}

\begin{theorem}\label{13.3.11.4.30}
Under the assumption of Theorem~\ref{12.12.1.12.39}
the coefficient $G_0$ satisfies
\begin{align}
G_{0}
\equiv 
G_0^0
+\langle \Psi_5,\cdot \rangle  \Psi_1^0
+\langle \Psi_1^0,\cdot \rangle  \Psi_5
\mod \mathcal B^0
\label{13.8.3.1.33}
\end{align}
and 
\begin{align}
HG_0=G_0H=1_{\mathcal L}.
\label{13.1.19.16.14}
\end{align}
Moreover, 
if $\widetilde{\mathcal E}_{\mathrm{qs}}=\{0\}$, then 
\begin{align}
\begin{split}
G_0
&=(1_{\mathcal L^*}+G_0^0V)^{-1}G_0^0
+\Delta^{\dagger}
\langle \Psi_5,\cdot\rangle \Psi_5\\
&=G_0^0(1_{\mathcal L}+VG_0^0)^{-1}
+\Delta^{\dagger}
\langle \Psi_5,\cdot\rangle \Psi_5,
\end{split}\label{13.3.12.6.25}
\intertext{and, if $\widetilde{\mathcal E}_{\mathrm{qs}}\neq \{0\}$, then }
\begin{split}
G_0
&=[\pi_1^*(1_{\mathcal L^*}+G_0^0V)\pi_2^*]^{-1}\pi_1^*G_0^0
+\langle \Psi_1^0,\cdot\rangle\Psi_5\\
&=G_0^0\pi_1[\pi_2(1_{\mathcal L}+VG_0^0)\pi_1]^{-1}
+\langle \Psi_5,\cdot\rangle\Psi_1^0,
\end{split}
\label{13.1.19.14.52b}
\end{align}
where $\pi_1$, $\pi_2$ are projections in $\mathcal L$ defined by
\begin{align}
\pi_1&=1_{\mathcal L}-\langle \Psi_1^0,\cdot\rangle V\Psi_5,
\quad
\pi_2=1_{\mathcal L}-\langle \Psi_5,\cdot\rangle V\Psi_1^0,
\label{13.8.15.4.50}
\end{align}
and $[\pi_1^*(1_{\mathcal L^*}+G_0^0V)\pi_2^*]^{-1}$, $[\pi_2(1_{\mathcal L}+VG_0^0)\pi_1]^{-1}$
are inverses in the projected spaces given by the Riesz--Schauder theorem.
\end{theorem}
\begin{proof}
Note that by Lemma~\ref{12.11.24.18.24} 
$G_0^0vm_0^\dagger v^*G_{0}^0\in \mathcal B^0$.
Then by the expression (\ref{13.1.19.16.23}) the formula (\ref{13.8.3.1.33}) follows.

We compute directly using the expression (\ref{13.1.19.16.23}):
\begin{align*}
\begin{split}
&(1_{\mathcal L^*}+G_0^0V)G_0\\
&=
G_0^0
+\gamma^2 G_{-1}^0vM_0v^*G_{-1}^0
-\gamma G_{-1}^0vv^*G_0^0
-\gamma G_0^0vv^*G_{-1}^0
-G_0^0vm_0^\dagger v^*G_{0}^0
\\&\phantom{={}}
+\gamma G_0^0vm_0^\dagger M_0v^*G_{-1}^0
+\gamma G_{-1}^0vM_0m_0^\dagger 
v^*G_{0}^0
-\gamma^2 G_{-1}^0vM_0m_0^\dagger M_0v^*G_{-1}^0
\\&\phantom{={}}
+\gamma G_0^0vv^*G_{-1}^0
+G_0^0vm_0^\dagger v^*G_{0}^0
-\gamma G_0^0vm_0^\dagger M_0v^*G_{-1}^0\\
&=
G_0^0
-\gamma G_{-1}^0v(1_{\mathcal K}-M_0m_0^\dagger) v^*G_0^0
+\gamma^2 G_{-1}^0v(M_0-M_0m_0^\dagger M_0)v^*G_{-1}^0.
\end{split}
\end{align*}
If we note $Q(M_0-M_0m_0^\dagger M_0)=0$, then we obtain
\begin{align}
&(1_{\mathcal L^*}+G_0^0V)G_0
=G_0^0
+\langle\Psi_5,\cdot\rangle\Psi_1^0.
\label{13.1.20.20.14}
\end{align}
Thus, it follows that $HG_0=1_{\mathcal L}$ by (\ref{13.1.20.20.14}) and Lemma~\ref{13.1.18.6.0},
and that $G_0H|_{\mathcal L}=1_{\mathcal L}$ by the adjoint argument.
These show (\ref{13.1.19.16.14}).

If $\widetilde{\mathcal E}_{\mathrm{qs}}= \{0\}$, then $1_{\mathcal L^*}+G_0^0V$ is invertible,
and (\ref{13.3.12.6.25}) can be easily verified by 
(\ref{13.1.20.20.14}) and (\ref{13.3.9.14.47}).
Now assume $\widetilde{\mathcal E}_{\mathrm{qs}}\neq \{0\}$.
We note, since $\langle \Psi_1^0,V\Psi_5\rangle=\langle \Psi_5,V\Psi_1^0\rangle =1\neq 0$,
the operators $\pi_1$ and $\pi_2$ defined by (\ref{13.8.15.4.50}) are in fact projections.
Repeating the arguments in the previous subsection, 
we can write using the Riesz--Schauder theorem that 
\begin{align*}
\pi_2^*G_0=[\pi_1^*(1_{\mathcal L^*}+G_0^0V)\pi_2^*]^{-1}\pi_1^*G_0^0.
\end{align*}
Hence, it suffices to compute $(1-\pi_2^*)G_0$.
Let us note by a direct computation with (\ref{13.1.19.16.23}) that 
\begin{align*}
VG_0
&= \langle \Psi_1^0,\cdot\rangle v\Phi_5+vm_0^\dagger v^*G_{0}^0.
\end{align*}
Then it follows that 
\begin{align*}
(1-\pi_2^*)G_0=\langle \Psi_1^0,\cdot \rangle \Psi_5,
\end{align*}
and hence we obtain the former identity of (\ref{13.1.19.14.52b}).
By the adjoint computation we can verify the latter, and we are done.
\end{proof}
\begin{remarks}\label{13.3.21.16.40}
\begin{enumerate}
\item
We take (\ref{13.1.19.16.23}) as the definition of $G_0$,
cf.\ Remark~\ref{13.3.11.5.22}.

\item
The action of $G_0H$ on the extended space $\mathbb C\mathbf n\oplus 
\mathbb C|\mathbf n|\oplus
\mathbb C\mathbf 1\oplus
\mathbb C\mbox{\boldmath$\sigma $}\oplus \mathcal L$
is more complicated than before.
\end{enumerate}
\end{remarks}

\subsection{The case: $P$, $m_0$ are not invertible, and $q_0$ is invertible}

\begin{theorem}\label{12.12.1.12.40}
Suppose that $\beta\ge 3$ in Assumption~\ref{12.11.8.1.9},
that 
$P$ and $m_0$ defined by (\ref{12.11.24.22.25}) and (\ref{12.11.27.8.11})
are not invertible in $\mathcal B(\mathcal K)$ and $\mathcal B(Q\mathcal K)$, respectively, 
and that $q_0$ defined by (\ref{12.12.1.1.50}) is invertible in 
$\mathcal B(S\mathcal K)$.
Let $N\in [-1,\beta-4]$ be any integer.
Then as $\kappa\to 0$ with $\real\kappa>0$,
the resolvent $R(\kappa)$ has the expansion in 
$\mathcal B^{N+3}$:
\begin{align}
R(\kappa)=\sum_{j=-1}^N\kappa^j G_j+\mathcal O(\kappa^{N+1}),\quad
G_j\in \mathcal B^{j+2},
\label{12.12.1.12.47}
\end{align}
and the coefficients $G_j$ can be computed explicitly.
\end{theorem}
\begin{proof}
By Proposition~\ref{12.11.13.19.5}, (\ref{12.11.11.12.35}) and (\ref{12.11.27.6.30}) 
we can write
\begin{align}
\begin{split}
R(\kappa)={}&R_0(\kappa)-\kappa R_0(\kappa)v(Q+\kappa M(\kappa))^{-1}v^*R_0(\kappa)\\
&-R_0(\kappa)v(Q+\kappa M(\kappa))^{-1}(S+m(\kappa))^\dagger
(Q+\kappa M(\kappa))^{-1}v^*R_0(\kappa)\\
&-\tfrac1{\kappa}R_0(\kappa)v(Q+\kappa M(\kappa))^{-1}
(S+m(\kappa))^\dagger\\
&\phantom{{}+{}}{}\cdot q(\kappa)^\dagger(S+m(\kappa))^\dagger(Q+\kappa M(\kappa))^{-1}v^*R_0(\kappa).
\end{split}\label{12.11.30.5.11b}
\end{align}
We compute the expansions of 
$(Q+\kappa M(\kappa))^{-1}$, $(S+m(\kappa))^\dagger$ and $q(\kappa)^\dagger$
employing the Neumann series.
Let $N\in [-1,\beta-4]$ be any integer.  
Then we have the expansions in 
$\mathcal B(\mathcal K)$, $\mathcal B(Q\mathcal K)$ and 
$\mathcal B(S\mathcal K)$:
\begin{align}
(Q+\kappa M(\kappa))^{-1}&
=\sum_{j=0}^{N+3}\kappa^j A_j+\mathcal O(\kappa^{N+4}),\quad
A_j\in\mathcal B(\mathcal K),\label{12.11.30.5.13b}\\
(S+m(\kappa))^\dagger&
=\sum_{j=0}^{N+2}\kappa^j B_j+\mathcal O(\kappa^{N+3}),\quad
B_j\in\mathcal B(Q\mathcal K),\label{12.11.30.5.14b}\\
q(\kappa)^\dagger
&=\sum_{j=0}^{N+1}\kappa^jC_j+\mathcal O(\kappa^{N+2}),\quad
C_j\in\mathcal B(S\mathcal K),\label{12.11.30.5.15b}
\end{align}
where
\begin{align*}
A_0&=Q+\gamma P,&
A_{j}&=\sum_{k=1}^{j}
\sum_{\genfrac{}{}{0pt}{}{j_1\ge 0,\dots,j_k\ge 0}{j_1+\cdots+j_k=j-k}}A_0
\prod_{l=1}^k(-M_{j_l}A_0)\quad\mbox{for } j\ge 1,
\\
B_0&=S+m_0^\dagger,&
B_{j}&=\sum_{k=1}^{j}
\sum_{\genfrac{}{}{0pt}{}{j_1\ge 1,\dots,j_k\ge 1}{j_1+\cdots+j_k=j}}B_0
\prod_{l=1}^k(-m_{j_l}B_0)\quad\mbox{for } j\ge 1,
\\
C_0&=q_0^\dagger,&
C_{j}&=\sum_{k=1}^{j}
\sum_{\genfrac{}{}{0pt}{}{j_1\ge 1,\dots,j_k\ge 1}{ j_1+\cdots+j_k=j}}C_0\prod_{l=1}^k(-q_{j_l}C_0)
\quad \mbox{for } j\ge 1,
\end{align*}
respectively.
We substitute (\ref{12.11.30.5.13b})--(\ref{12.11.30.5.15b}) and
\begin{align*}
R_0(\kappa)=\sum_{j=-1}^{N+1}G_j^0+\mathcal O(\kappa^{N+2})\quad\mbox{ in }
\mathcal B^{N+3}
\end{align*}
into (\ref{12.11.30.5.11b}), and then obtain
in the topology of $\mathcal B^{N+3}$
\begin{align*}
R(\kappa)=\sum_{j=-1}^{N}\kappa^jG_j+\mathcal O(\kappa^{N+1}),
\end{align*}
where for $j\ge -1$
\begin{align}
\begin{split}
G_j={}&
G_j^0
-\sum_{\genfrac{}{}{0pt}{}{j_1\ge -1, j_2\ge 0,j_3\ge -1}{j_1+j_2+j_3=j-1}}G_{j_1}^0vA_{j_2}v^*G_{j_3}^0
\\&
-\sum_{\genfrac{}{}{0pt}{}{j_1\ge -1, j_2\ge 0,j_3\ge 0,j_4\ge 0,j_5\ge -1}{j_1+j_2+j_3+j_4+j_5=j}}
G_{j_1}^0vA_{j_2}B_{j_3}A_{j_4}v^*G_{j_5}^0
\\&
-\sum_{\genfrac{}{}{0pt}{}{j_1\ge -1, j_2\ge 0,\dots,j_6\ge 0,j_7\ge -1}{ j_1+\dots+j_7=j+1}}
G_{j_1}^0vA_{j_2}B_{j_3}C_{j_4}B_{j_5}A_{j_6}v^*G_{j_7}^0.
\end{split}\label{12.12.1.12.46}
\end{align}
Here, as in the proof of Theorem~\ref{12.12.1.12.39}, 
we have used the identities
(\ref{12.11.30.6.34}) to guarantee that the error term is $\mathcal O(\kappa^{N+1})$
and also that the coefficient of $\kappa^{-2}$ vanishes.
By (\ref{12.12.1.12.46}) and (\ref{12.11.30.6.34}) again we obtain
$G_j\in \mathcal B^{j+2}$
and hence the expansion (\ref{12.12.1.12.47}) is verified.
\end{proof}

By (\ref{12.12.1.12.46}) and (\ref{12.11.30.6.34})
we can write
\begin{align}
\begin{split}
G_{-1}&=G_{-1}^0-G_{-1}^0vA_0v^*G_{-1}^0\\
&\phantom{={}}-(G_{-1}^0vA_1+G_0^0vA_0)B_0C_0B_0(A_1v^*G_{-1}^0+A_0v^*G_0^0).
\end{split}\label{13.3.11.6.7}\\
\begin{split}
G_{0}&=G_{0}^0-G_{-1}^0vA_1v^*G_{-1}^0-G_{-1}^0vA_0v^*G_{0}^0-G_{0}^0vA_0v^*G_{-1}^0\\
&\phantom{={}}-(G_{-1}^0vA_1+G_0^0vA_0)B_0(A_1v^*G_{-1}^0+A_0v^*G_0^0)\\
&\phantom{={}}-(G_{-1}^0vA_1+G_0^0vA_0)B_0C_1B_0(A_1v^*G_{-1}^0+A_0v^*G_0^0)\\
&\phantom{={}}-(G_{-1}^0vA_1+G_0^0vA_0)B_0C_0B_1(A_1v^*G_{-1}^0+A_0v^*G_0^0)\\
&\phantom{={}}-(G_{-1}^0vA_1+G_0^0vA_0)B_1C_0B_0(A_1v^*G_{-1}^0+A_0v^*G_0^0)\\
&\phantom{={}}-(G_{-1}^0vA_1+G_0^0vA_0)B_0C_0B_0
(A_2v^*G_{-1}^0+A_1v^*G_0^0+A_0v^*G_{1}^0)\\
&\phantom{={}}-(G_{-1}^0vA_2+G_0^0vA_1+G_1^0vA_0)B_0C_0B_0
(A_1v^*G_{-1}^0+A_0v^*G_0^0).
\end{split}\label{13.3.11.6.8}
\end{align}
We rewrite the expressions (\ref{13.3.11.6.7})
and (\ref{13.3.11.6.8}) substituting
\begin{align*}
A_0&=Q+\gamma P,\quad
A_1=-(Q+\gamma P)M_0(Q+\gamma P),\\
A_2&=-(Q+\gamma P)M_1(Q+\gamma P)+(Q+\gamma P)M_0(Q+\gamma P)M_0(Q+\gamma P),\\
B_0&=S+m_0^\dagger,\quad
B_1=-(S+m_0^\dagger)m_1(S+m_0^\dagger),\quad
C_0=q_0^\dagger,\quad
C_1=-q_0^\dagger q_1q_0^\dagger.
\end{align*}
We note that 
$SM_0Q=QM_0S=0$, $q_0=Sm_1S$ and $z=\gamma G_{-1}^0vM_0-G_0^0v$ on $S\mathcal K$.
Then after some computations,
\begin{align}
G_{-1}
&=G_{-1}^0-\gamma G_{-1}^0vv^*G_{-1}^0-zq_0^{\dagger}z^*,\label{13.3.13.3.2}\\
\begin{split}
G_{0}&=G_{0}^0+\gamma^2 G_{-1}^0vM_0v^*G_{-1}^0
-\gamma G_{-1}^0vv^*G_{0}^0-\gamma G_{0}^0vv^*G_{-1}^0\\
&\phantom{={}}-(\gamma G_{-1}^0vM_0-G_0^0v)m_0^\dagger(\gamma M_0v^*G_{-1}^0-v^*G_0^0)
+zSz^*
+z q_0^\dagger q_1q_0^\dagger z^*\\
&\phantom{={}}
+z q_0^\dagger m_1m_0^\dagger(\gamma M_0v^*G_{-1}^0-v^*G_0^0)
+(\gamma G_{-1}^0vM_0-G_0^0v)m_0^\dagger m_1q_0^\dagger z^*
\\
&\phantom{={}}
+zq_0^\dagger v^*G_{1}^0(1_{\mathcal L}-\gamma vv^*G_{-1}^0)
+\gamma zq_0^\dagger M_0P(\gamma M_0v^*G_{-1}^0-v^*G_0^0)\\
&\phantom{={}}
+(1_{\mathcal L^*}-\gamma G_{-1}^0vv^*) G_1^0vq_0^\dagger z^*
+\gamma (\gamma G_{-1}^0vM_0-G_0^0v)PM_0q_0^\dagger z^*.
\end{split}\label{13.3.13.3.3}
\end{align}
In the proofs of the following theorems we need an explicit expression of $q_0^\dagger$.
Let us compute it here.
By (\ref{12.12.19.8.41}) we have
\begin{align*}
q_0=&-\tfrac{1}2\bigl[
\bigl(1
+|\langle \Phi_3^*,\Phi_2\rangle|^2 \bigr) \langle \Phi_3,\cdot\rangle \Phi_3
+\langle \Phi_4,\cdot \rangle\Phi_4\\
&\phantom{{}-\tfrac{1}2\bigl[}
+\overline{\langle \Phi_3^*,\Phi_2\rangle} \langle \Phi_3,\cdot \rangle\Phi_4
+\langle \Phi_3^*,\Phi_2\rangle \langle\Phi_4,\cdot \rangle\Phi_3\bigr].
\end{align*}
Since $\Phi_3$ and $\Phi_4$ are orthogonal, we can compute the inverse $q_0^\dagger$,
using the $2\times 2$ matrix inverse if $\Phi_3\neq 0$ and $\Phi_4\neq 0$,
and the $1\times 1$ matrix inverse otherwise.
At last after some computations we have the following two cases:
If $\Phi_4\neq 0$, then
\begin{align}
\begin{split}
q_0^\dagger
&=-2\bigl[\langle \Phi_3^*,\cdot\rangle \Phi_3^*
+\bigl(1+|\langle \Phi_3^*,\Phi_2\rangle|^2 \bigr) \langle \Phi_4^*,\cdot \rangle\Phi_4^*
\\&\phantom{={}-2\bigl[}
-\overline{\langle \Phi_3^*,\Phi_2\rangle}\langle \Phi_3^*,\cdot \rangle\Phi_4^*
-\langle \Phi_3^*,\Phi_2\rangle\langle\Phi_4^*,\cdot \rangle\Phi_3^*\bigr],
\end{split}\label{13.8.3.3.33}
\intertext{and, if $\Phi_4=0$, then}
q_0^\dagger&=-2\bigl(1+|\langle \Phi_3^*,\Phi_2\rangle|^2 \bigr)^{-1}
\langle \Phi_3^*,\cdot\rangle \Phi_3^*.\label{13.8.3.3.34}
\end{align}

\begin{theorem}\label{13.3.11.5.17}
Under the assumption of Theorem~\ref{12.12.1.12.40}
the coefficient $G_{-1}$ is of the form (\ref{13.3.11.2.46})
(up to the factor $i$ due to a different convention).
\end{theorem}
\begin{proof}
The assertion is obvious by (\ref{13.3.13.3.2}),
(\ref{13.8.3.3.33}), (\ref{13.8.3.3.34}), Tables~\ref{13.3.6.0.28a} and \ref{13.3.6.0.28}.
\end{proof}

\begin{theorem}\label{13.3.11.4.31}
Under the assumption of Theorem~\ref{12.12.1.12.40}
the coefficient $G_0$ satisfies
\begin{align}
HG_0=G_0H=1_{\mathcal L}.
\label{13.1.19.16.15}
\end{align}
Moreover, if $\widetilde{\mathcal E}_{\mathrm{qs}}=\{0\}$, then
there exists a computable constant $\Delta_1\in \mathbb R$ and $\Delta_2\in\mathbb C$ such that 
\begin{align}
\begin{split}
G_0
&=(1_{\mathcal L^*}+G_0^0V)^{-1}G_0^0
+\Delta_1\langle \Psi_3,\cdot\rangle \Psi_3
+\Delta_2\langle \Psi_3,\cdot\rangle \Psi_6
+\overline{\Delta}_2\langle \Psi_6,\cdot\rangle \Psi_3\\
&=G_0^0(1_{\mathcal L}+VG_0^0)^{-1}
+\Delta_1\langle \Psi_3,\cdot\rangle \Psi_3
+\Delta_2\langle \Psi_3,\cdot\rangle \Psi_6
+\overline{\Delta}_2\langle \Psi_6,\cdot\rangle \Psi_3.
\end{split}\label{13.3.12.6.25a}
\end{align}
\end{theorem}
\begin{proof}
Let us compute $(1_{\mathcal L^*}+G_0^0V)G_0$ using (\ref{13.3.13.3.3}).
Here we omit the details, but describe the outline.
We first use the identities
\begin{align}
(1_{\mathcal L^*}+G_0^0V)(\gamma G_{-1}^0vM_0-G_0^0v)
&=
-G_0^0vUQM_0+\gamma G_{-1}^0vM_0,\label{13.8.3.4.41}\\
(1_{\mathcal L^*}+G_0^0V)zS
&=\gamma G_{-1}^0vM_0S,\label{13.8.3.4.42}\\
(1_{\mathcal L^*}+G_0^0V)(1_{\mathcal L^*}-\gamma G_{-1}^0vv^*)
&=
1_{\mathcal L^*}-\gamma G_{-1}^0vv^*+G_0^0vUQv^*.\label{13.8.3.4.43}
\end{align}
Then the contributions from the first five terms of (\ref{13.3.13.3.3})
can be directly computed.
As for the contributions from the remaining eight terms,
we use Lemma~\ref{12.12.19.4.4}, e.g., $m_0=QM_0Q$,
$m_0m_0^\dagger=Q-S$, 
$m_1=QM_1Q-QM_0(Q+\gamma P)M_0Q$, $Sm_0=0$, $q_0=Sm_1S$ and so on.
At last we obtain
\begin{align}
(1_{\mathcal L^*}+G_0^0V)G_{0}
&=G_0^0
+\gamma G_{-1}^0v(1_{\mathcal K}-M_0m_0^\dagger) (\gamma M_0v^*G_{-1}^0-v^*G_0^0)\notag\\
&\quad+\gamma G_{-1}^0vM_0Sz^*
+\gamma G_{-1}^0vM_0 q_0^\dagger q_1q_0^\dagger z^*\notag\\
&\quad+\gamma G_{-1}^0vM_0q_0^\dagger m_1m_0^\dagger(\gamma M_0v^*G_{-1}^0-v^*G_0^0)\notag\\
&\quad+\gamma G_{-1}^0vM_0m_0^\dagger m_1q_0^\dagger z^*
+\gamma G_{-1}^0vM_0q_0^\dagger v^*G_{1}^0(1_{\mathcal L}-\gamma vv^*G_{-1}^0)\notag\\
&\quad
+\gamma^2 G_{-1}^0vM_0q_0^\dagger M_0P(\gamma M_0v^*G_{-1}^0-v^*G_0^0)\notag\\
&\quad
+(1_{\mathcal L^*}-\gamma G_{-1}^0v^*v)G_1^0v q_0^\dagger z^*
+\gamma^2 G_{-1}^0vM_0 PM_0q_0^\dagger z^*.
\label{13.3.14.1.55}
\end{align}
The above computation can be followed rather easier if we focus on the terms with the factor $G_0^0$ to their left.
Hence by (\ref{13.3.14.1.55}) and Lemma~\ref{13.1.18.6.0}
we can deduce (\ref{13.1.19.16.15}).

Next, assume that $1_{\mathcal L^*}+G_0^0V$ is invertible.
In this case we further compute 
$(1_{\mathcal L^*}+G_0^0V)G_{0}(1_{\mathcal L}+VG_0^0)$.
We use (\ref{13.3.14.1.55}) and the adjoints of (\ref{13.8.3.4.41})--(\ref{13.8.3.4.43}),
and proceed similarly to the above.
After some computations, we obtain
\begin{align*}
(1_{\mathcal L^*}+G_0^0V)&G_{0}(1_{\mathcal L}+VG_0^0)
=
G_0^0+G_0^0VG_0^0
+\gamma^2 G_{-1}^0v(M_0-M_0m_0^\dagger M_0) v^*G_{-1}^0
\\&
+\gamma^2 G_{-1}^0vM_0S  M_0v^*G_{-1}^0
+\gamma^2 G_{-1}^0vM_0 q_0^\dagger q_1q_0^\dagger  M_0v^*G_{-1}^0
\\&
+\gamma^2 G_{-1}^0vM_0q_0^\dagger m_1m_0^\dagger  M_0v^*G_{-1}^0
+\gamma^2 G_{-1}^0vM_0m_0^\dagger m_1q_0^\dagger  M_0v^*G_{-1}^0
\\&
+\gamma^3 G_{-1}^0vM_0q_0^\dagger M_0P  M_0v^*G_{-1}^0
+\gamma^3 G_{-1}^0vM_0 PM_0q_0^\dagger  M_0v^*G_{-1}^0
\\&
+\gamma G_{-1}^0vM_0q_0^\dagger v^*G_{1}^0(1_{\mathcal L}-\gamma vv^*G_{-1}^0)
\\
&
+(1_{\mathcal L^*}-\gamma G_{-1}^0v^*v)G_1^0v q_0^\dagger (\gamma M_0v^*G_{-1}^0).
\end{align*}
Hence we can write for some $\alpha\in\mathbb R$ and $\beta\in\mathbb C$
\begin{align*}
(1_{\mathcal L^*}+G_0^0V)G_{0}(1_{\mathcal L}+VG_0^0)
&=
G_0^0+G_0^0VG_0^0\\
&\quad+\alpha\langle \Psi_1^0,\cdot\rangle \Psi_1^0
+\beta\langle \Psi_1^0,\cdot\rangle \Psi_2^0
+\overline{\beta}\langle \Psi_2^0,\cdot\rangle \Psi_1^0.
\end{align*}
We note that by Tables~\ref{13.3.6.0.28a} and \ref{13.3.6.0.28}
we necessarily have $\Phi_4=0$.
Then by (\ref{13.3.9.14.47}) we can verify obtain (\ref{13.3.12.6.25a}).
\end{proof}
\begin{remark}\label{13.8.18.5.25}
Remarks similar to Remarks~\ref{13.3.21.16.40}
hold also for this case.
By (\ref{13.3.13.3.3}) we have 
\begin{align*}
G_{0}-G_{0}^0&\equiv
-\gamma G_{-1}^0vv^*G_{0}^0
+zq_0^\dagger v^*G_{1}^0
-\gamma zq_0^\dagger M_0 v^*G_0^0+\mathrm{h.c.} \mod \mathcal B^0,
\end{align*}
and hence, using (\ref{13.8.3.3.33}) and (\ref{13.8.3.3.34}),
we can express $G_{0}-G_{0}^0$ modulo $\mathcal B^0$ as in the former subsections.
However, the expression seems very long and we do not elaborate it here.
We also note that we do not investigate the case $\widetilde{\mathcal E}_{\mathrm{qs}}\neq \{0\}$.
\end{remark}

\subsection{The case: $P$, $m_0$, $q_0$ are not invertible}

\begin{theorem}\label{12.12.1.12.41}
Suppose that $\beta\ge 4$ in Assumption~\ref{12.11.8.1.9},
and that 
$P$, $m_0$ and $q_0$ defined by (\ref{12.11.24.22.25}), (\ref{12.11.27.8.11})
and (\ref{12.12.1.1.50}) are not invertible in $\mathcal B(\mathcal K)$,
$\mathcal B(Q\mathcal K)$ and $\mathcal B(S\mathcal K)$, respectively.
Let $N\in [-2,\beta-6]$ be any integer.
Then as $\kappa\to 0$ with $\real\kappa>0$,
the resolvent $R(\kappa)$ has the expansion in 
$\mathcal B^{N+4}$:
\begin{align}
R(\kappa)=\sum_{j=-2}^N\kappa^j G_j+\mathcal O(\kappa^{N+1}),\quad
G_j\in \mathcal B^{j+3},
\label{12.12.1.20.44}
\end{align}
and the coefficients $G_j$ can be computed explicitly.
\end{theorem}
\begin{proof}
By Proposition~\ref{12.11.13.19.5}, (\ref{12.11.11.12.35}), (\ref{12.11.27.6.30}) 
and (\ref{12.11.27.6.43}) we can write
\begin{align}
R(\kappa)&=R_0(\kappa)-\kappa R_0(\kappa)v(Q+\kappa M(\kappa))^{-1}v^*R_0(\kappa)\notag\\
&\quad-R_0(\kappa)v(Q+\kappa M(\kappa))^{-1}Q(S+m(\kappa))^\dagger
(Q+\kappa M(\kappa))^{-1}v^*R_0(\kappa)\notag\\
&\quad-\tfrac1{\kappa}R_0(\kappa)v(Q+\kappa M(\kappa))^{-1}
(S+m(\kappa))^\dagger\notag\\
&\qquad\cdot (T+q(\kappa))^\dagger(S+m(\kappa))^\dagger(Q+\kappa M(\kappa))^{-1}v^*R_0(\kappa)\notag\\
&\quad-\tfrac1{\kappa^2}R_0(\kappa)v(Q+\kappa M(\kappa))^{-1}Q(S+m(\kappa))^\dagger(T+q(\kappa))^\dagger\notag\\
&\qquad\cdot r(\kappa)^\dagger(T+q(\kappa))^\dagger
(S+m(\kappa))^\dagger(Q+\kappa M(\kappa))^{-1}v^*R_0(\kappa).
\label{b12.11.30.5.11b}
\end{align}
We compute the expansions of $(Q+\kappa M(\kappa))^{-1}$, $(S+m(\kappa))^\dagger$,
$(T+q(\kappa))^\dagger$ and $r(\kappa)^\dagger$
employing the Neumann series.
Let $N\in [-2,\beta-6]$ be any integer.
Then we have the expansions in 
$\mathcal B(\mathcal K)$, $\mathcal B(Q\mathcal K)$ and $\mathcal B(S\mathcal K)$:
\begin{align}
(Q+\kappa M(\kappa))^{-1}&
=\sum_{j=0}^{N+5}\kappa^j A_j+\mathcal O(\kappa^{N+6}),\quad
A_j\in\mathcal B(\mathcal K),\label{b12.11.30.5.13b}\\
(S+m(\kappa))^\dagger&
=\sum_{j=0}^{N+4}\kappa^j B_j+\mathcal O(\kappa^{N+5}),\quad
B_j\in\mathcal B(Q\mathcal K),\label{b12.11.30.5.14b}\\
(T+q(\kappa))^\dagger
&=\sum_{j=0}^{N+3}\kappa^jC_j+\mathcal O(\kappa^{N+4}),\quad
C_j\in\mathcal B(S\mathcal K),\label{b12.11.30.5.15b}\\
r(\kappa)^\dagger
&=\sum_{j=0}^{N+2}\kappa^jD_j+\mathcal O(\kappa^{N+3}),\quad
D_j\in\mathcal B(T\mathcal K),\label{b12.11.30.5.16b}
\end{align}
\begin{align*}
A_0&=Q+\gamma P,&
A_{j}&=\sum_{k=1}^{j}
\sum_{\genfrac{}{}{0pt}{}{j_1\ge 0,\dots,j_k\ge 0}{j_1+\cdots+j_k=j-k}}A_0
\prod_{l=1}^k(-M_{j_l}A_0)\quad\mbox{for } j\ge 1,
\\
B_0&=S+m_0^\dagger,&
B_{j}&=\sum_{k=1}^{j}
\sum_{\genfrac{}{}{0pt}{}{j_1\ge 1,\dots,j_k\ge 1}{j_1+\cdots+j_k=j}}B_0
\prod_{l=1}^k(-m_{j_l}B_0)\quad\mbox{for } j\ge 1,
\\
C_0&=T+q_0^\dagger,&
C_{j}&=\sum_{k=1}^{j}
\sum_{\genfrac{}{}{0pt}{}{j_1\ge 1,\dots,j_k\ge 1}{j_1+\cdots+j_k=j}}C_0\prod_{l=1}^k(-q_{j_l}C_0)
\quad \mbox{for } j\ge 1,
\\
D_0&=r_0^\dagger,&
D_{j}&=\sum_{k=1}^{j}
\sum_{\genfrac{}{}{0pt}{}{j_1\ge 1,\dots,j_k\ge 1}{j_1+\cdots+j_k=j}}D_0\prod_{l=1}^k(-r_{j_l}D_0)
\quad \mbox{for } j\ge 1,
\end{align*}
respectively.
We substitute the expansions (\ref{b12.11.30.5.13b})--(\ref{b12.11.30.5.16b})
and
\begin{align*}
R_0(\kappa)=\sum_{j=-1}^{N+2}G_j^0+\mathcal O(\kappa^{N+3})\quad\mbox{ in }
\mathcal B^{N+4}
\end{align*}
into (\ref{b12.11.30.5.11b}), and then obtain
in the topology of $\mathcal B^{N+4}$
\begin{align*}
R(\kappa)=\sum_{j=-2}^{N}\kappa^jG_j+\mathcal O(\kappa^{N+1}),
\end{align*}
where for $j\ge -2$ with $G_{-2}^0=0$
\begin{align}
\begin{split}
G_j={}&
G_j^0
-\sum_{\genfrac{}{}{0pt}{}{j_1\ge -1, j_2\ge 0,j_3\ge -1}{j_1+j_2+j_3=j-1}}G_{j_1}^0vA_{j_2}v^*G_{j_3}^0
\\&
-\sum_{\genfrac{}{}{0pt}{}{j_1\ge -1, j_2\ge 0,j_3\ge 0,j_4\ge 0,j_5\ge -1}{j_1+j_2+j_3+j_4+j_5=j}}
G_{j_1}^0vA_{j_2}B_{j_3}A_{j_4}v^*G_{j_5}^0
\\&
-\sum_{\genfrac{}{}{0pt}{}{j_1\ge -1, j_2\ge 0,\dots,j_6\ge 0,j_7\ge -1}{j_1+\dots+j_7=j+1}}
G_{j_1}^0vA_{j_2}B_{j_3}C_{j_4}B_{j_5}A_{j_6}v^*G_{j_7}^0.
\\&
-\sum_{\genfrac{}{}{0pt}{}{j_1\ge -1, j_2\ge 0,\dots,j_8\ge 0,j_9\ge -1}{j_1+\dots+j_9=j+2}}
G_{j_1}^0vA_{j_2}B_{j_3}C_{j_4}D_{j_5}C_{j_6}B_{j_7}A_{j_8}v^*G_{j_9}^0.
\end{split}\label{b12.12.1.12.46}
\end{align}
As in the proof of Theorem~\ref{12.12.1.12.39}, we have used the
identities in (\ref{12.11.30.6.34}) to guarantee that the error term is $\mathcal O(\kappa^{N+1})$.
By (\ref{b12.12.1.12.46}) and (\ref{12.11.30.6.34})  we obtain
$G_j\in \mathcal B^{j+3}$, 
and hence the expansion (\ref{12.12.1.20.44}) is obtained.
\end{proof}

We have by (\ref{b12.12.1.12.46}) and (\ref{12.11.30.6.34})
\begin{align}
G_{-2}&=
-(G_{-1}^0vA_1+G_0^0vA_0)B_0C_0D_0C_0B_0(A_1v^*G_{-1}^0+A_0v^*G_0^0),\label{13.8.4.21.47}\\
\begin{split}
G_{-1}&=G_{-1}^{0}-G_{-1}^0vA_0v^*G_{-1}^0
\\&\phantom{={}}
-(G_{-1}^0vA_1+G_0^0vA_0)B_0C_0B_0(A_1v^*G_{-1}^0+A_0v^*G_0^0)
\\
&\phantom{={}}-(G_{-1}^0vA_1+G_0^0vA_0)(B_1C_0D_0C_0B_0+B_0C_1D_0C_0B_0
+B_0C_0D_1C_0B_0
\\&\phantom{={}-{}}
+B_0C_0D_0C_1B_0
+B_0C_0D_0C_0B_1)(A_1v^*G_{-1}^0+A_0v^*G_0^0)\\
&\phantom{={}}-(G_{-1}^0vA_1+G_0^0vA_0)B_0C_0D_0C_0B_0
(A_0v^*G_{1}^0+A_1v^*G_0^0+A_2v^*G_{-1}^0)\\
&\phantom{={}}-(G_1^0vA_0+G_0^0vA_1+G_{-1}^0vA_2)B_0C_0D_0C_0B_0
(A_1v^*G_{-1}^0+A_0v^*G_0^0).
\end{split}\label{13.8.4.21.48}
\end{align}
The coefficient $G_0$ will be considered separately.
We proceed using Lemma~\ref{12.12.19.4.4} and 
\begin{align}
\begin{split}
A_0&=Q+\gamma P,\quad
A_1=-(Q+\gamma P)M_0(Q+\gamma P),\\
A_2&=-(Q+\gamma P)M_1(Q+\gamma P)+(Q+\gamma P)M_0(Q+\gamma P)M_0(Q+\gamma P),\\
B_0&=S+m_0^\dagger,\quad
B_1=-(S+m_0^\dagger)m_1(S+m_0^\dagger),\\
C_0&=T+q_0^\dagger,\quad
C_1=-(T+q_0^\dagger)q_1(T+q_0^\dagger),\quad
D_0=r_0^\dagger,\quad
D_1=-r_0^\dagger r_1r_0^\dagger.
\end{split}\label{13.3.24.19.55}
\end{align}
First we substitute (\ref{13.3.24.19.55}) to (\ref{13.8.4.21.47}) and (\ref{13.8.4.21.48}),
and then we compute it particularly noting 
\begin{align*}
(G_{-1}^0vA_1+G_0^0vA_0)S=-zS,\quad
Tm_1=0,\quad
TQM_1Q=0,\quad
TM_0=0.
\end{align*}
After some computations we obtain the expressions
\begin{align}
G_{-2}&=
-zr_0^\dagger z^*,\label{13.8.4.20.31}\\
\begin{split}
G_{-1}&=G_{-1}^{0}-\gamma G_{-1}^0vv^*G_{-1}^0
+z(T+r_0^\dagger r_1r_0^\dagger
)z^*
+z(-q_0^\dagger+q_0^\dagger q_1r_0^\dagger
+r_0^\dagger q_1 q_0^\dagger)z^*
\\&\phantom{={}}
+zr_0^\dagger v^*G_{1}^0(1-\gamma vv^*G_{-1}^0)
+(1-\gamma G_{-1}^0vv^*)G_1^0vr_0^\dagger z^*.
\end{split}\label{13.8.4.20.32}
\end{align}

\begin{theorem}\label{13.3.11.5.18}
Under the assumption of Theorem~\ref{12.12.1.12.41}
the coefficients $G_{-2}$ and $G_{-1}$ are of the form (\ref{13.3.11.2.46})
(up to the factors $-1$ and $i$, respectively, 
due to different conventions).
\end{theorem}
\begin{proof}
By (\ref{13.8.4.20.31}) and (\ref{13.3.8.16.2}) we have
\begin{align*}
G_{-2}&=zT(Tz^*zT)^\dagger Tz^*.
\end{align*}
This implies that $G_{-2}$ on $\mathcal H$
is the orthogonal projection onto the eigenspace $E$,
and hence the assertion for $G_{-2}$ follows.

We rewrite (\ref{13.8.4.20.32}) modulo $\langle E,\cdot \rangle E$ as follows:
If $\Phi_1\neq 0$, then 
\begin{align*}
G_{-1}\equiv{}& -z(S- r_0^\dagger q_1)q_0^\dagger (S-q_1r_0^\dagger)z^*,
\intertext{ and, if $\Phi_1= 0$, then}
G_{-1}\equiv {}&
-z(S- r_0^\dagger q_1)q_0^\dagger (S-q_1r_0^\dagger)z^*\\
&+\tfrac12 \bigl\langle\bigl(\Psi_1^0+\tfrac12zr_0^\dagger v^*\mathbf n^2\bigr),\cdot\bigr\rangle
\bigl(\Psi_1^0+\tfrac12zr_0^\dagger v^*\mathbf n^2\bigr)
\end{align*}
where $\mathbf n^2\in (\mathcal L^2)^*$ is 
the sequence whose $n$-th entry is $n^2$.
By (\ref{13.8.3.3.33}), (\ref{13.8.3.3.34}), Tables~\ref{13.3.6.0.29a} and \ref{13.3.6.0.29}
we can see that the assertion for $G_{-1}$ holds.
\end{proof}

Finally we consider $G_0$ with $\beta\ge 5$.
The expression is very long, and we should proceed omitting the details.
By (\ref{b12.12.1.12.46}) we first have, 
removing unnecessary factors from the beginning, 
\begin{align}
\begin{split}
G_0
&
=G_0^{0}-G_{-1}^0vA_1v^*G_{-1}^0
-G_{-1}^0vA_0v^*G_{0}^0-G_0^0vA_0v^*G_{-1}^0
\\&\phantom{={}}
-(G_{-1}^0vA_1+G_0^0v)B_0(A_1v^*G_{-1}^0+v^*G_0^0)
\\&\phantom{={}}
-(G_{-1}^0vA_1+G_0^0v)(B_1C_0+C_1+C_0B_1)(A_1v^*G_{-1}^0+A_0v^*G_0^0)
\\&\phantom{={}}
-(G_{-1}^0vA_1+G_0^0v)C_0(A_2v^*G_{-1}^0+A_1v^*G_0^0+v^*G_{1}^0)
\\&\phantom{={}}
-(G_{-1}^0vA_2+G_0^0vA_1+G_1^0v)C_0(A_1v^*G_{-1}^0+v^*G_0^0)
\\&\phantom{={}}
-(G_{-1}^0vA_1+G_0^0v)(B_2D_0+C_2D_0+D_2+D_0C_2+D_0B_2
\\&\phantom{={}-{}}
+B_1C_1D_0+B_1D_1+B_1D_0C_1+B_1D_0B_1+C_1D_1+C_1D_0C_1
\\&\phantom{={}-{}}
+C_1D_0B_1+D_1C_1+D_1B_1+D_0C_1B_1)(A_1v^*G_{-1}^0+v^*G_0^0)
\\&\phantom{={}}
-(G_{-1}^0vA_1+G_0^0v)(B_1D_0+C_1D_0+D_1
\\&\phantom{={}-{}}
+D_0C_1+D_0B_1)(A_2v^*G_{-1}^0+A_1v^*G_0^0+v^*G_{1}^0)
\\&\phantom{={}}
-(G_{-1}^0vA_2+G_0^0vA_1+G_1^0v)(B_1D_0+C_1D_0+D_1
\\&\phantom{={}-{}}
+D_0C_1+D_0B_1)(A_1v^*G_{-1}^0+v^*G_0^0)
\\&\phantom{={}}
-(G_{-1}^0vA_1+G_0^0v)D_0(A_3v^*G_{-1}^0+A_2v^*G_{0}^0+A_1v^*G_1^0+v^*G_{2}^0)
\\&\phantom{={}}
-(G_{-1}^0vA_2+G_0^0vA_1+G_1^0v)D_0(A_2v^*G_{-1}^0+A_1v^*G_0^0+v^*G_{1}^0)
\\&\phantom{={}}
-(G_{-1}^0vA_3+G_{0}^0vA_2+G_1^0vA_1+G_2^0v)D_0(A_1v^*G_{-1}^0+v^*G_0^0).
\end{split}
\label{13.3.24.22.7}
\end{align}
We are going to compute $HG_0$ using the expressions in Propositions~\ref{12.11.9.6.23}, 
\ref{12.11.27.11.18}, Lemma~\ref{12.12.19.4.4}, (\ref{13.3.24.19.55})
and 
\begin{align}
\begin{split}
A_3
&=-(Q+\gamma P) M_2(Q+\gamma P)+(Q+\gamma P)M_0(Q+\gamma P)M_1(Q+\gamma P)\\
&\phantom{={}}+(Q+\gamma P)M_1(Q+\gamma P)M_0(Q+\gamma P)\\
&\phantom{={}}-(Q+\gamma P)M_0(Q+\gamma P)M_0(Q+\gamma P)M_0(Q+\gamma P),\\
B_2&=-(S+m_0^\dagger)m_2(S+m_0^\dagger)
+(S+m_0^\dagger)m_1(S+m_0^\dagger)m_1(S+m_0^\dagger).
\end{split}
\label{13.8.5.13.41}
\end{align}
On the other hand we will not need expressions of $C_2$ and $D_2$.
We will use only $C_2\in \mathcal B(S\mathcal K)$ and $D_2\in\mathcal B(T\mathcal K)$.
After some computations employing these,
\begin{align}
\begin{split}
&
H(G_0^{0}-G_{-1}^0vA_1v^*G_{-1}^0-G_{-1}^0vA_0v^*G_{0}^0-G_0^0vA_0v^*G_{-1}^0)
\\&
=1_{\mathcal L}+vUQ(A_1v^*G_{-1}^0+v^*G_0^{0}),
\\&
H(G_{-1}^0vA_1+G_0^0v)Q=vUm_0,
\\&
H(G_{-1}^0vA_2+G_0^0vA_1+G_1^0v)Q=vUm_1,
\\&
H(G_{-1}^0vA_3+G_{0}^0vA_2+G_1^0vA_1+G_2^0v)T=(vUm_2-G_0^0v)T,
\end{split}\label{13.8.5.15.2}
\end{align}
where we used also Lemma~\ref{12.12.27.15.3}, the identities $PM_0T=0$ and $QM_0T=0$.
By (\ref{13.8.5.15.2}) we obtain the following slightly simplified formula for $HG_0$.
Here we just gather the terms with the same factors without expanding the parentheses
except for the contribution from the last term of (\ref{13.3.24.22.7}):
\begin{align}
\begin{split}
HG_0
&=
1_{\mathcal L}+G_0^0vD_0(A_1v^*G_{-1}^0+v^*G_0^0)
\\&\phantom{={}}
+vU\Bigl[Q-m_0B_0-m_0(B_1C_0+C_1+C_0B_1)-m_1C_0
\\&\phantom{={}+vU\Bigl[}
-m_0(B_2D_0+C_2D_0+D_2+D_0C_2+D_0B_2
\\&\phantom{={}+vU\Bigl[-m_0(}
+B_1C_1D_0+B_1D_1+B_1D_0C_1+B_1D_0B_1+C_1D_1
\\&\phantom{={}+vU\Bigl[-m_0(}
+C_1D_0C_1+C_1D_0B_1+D_1C_1+D_1B_1+D_0C_1B_1)
\\&\phantom{={}+vU\Bigl[}
-m_1(B_1D_0+C_1D_0+D_1+D_0C_1+D_0B_1)
-m_2D_0
\Bigr]
\\&\phantom{={}+{}}\cdot
(A_1v^*G_{-1}^0+v^*G_0^{0})
\\&\phantom{={}}
-vU\Bigl[m_0C_0+m_0(B_1D_0+C_1D_0+D_1+D_0C_1+D_0B_1)
+m_1D_0
\Bigr]
\\&\phantom{={}-{}}\cdot
(A_2v^*G_{-1}^0+A_1v^*G_0^0+v^*G_{1}^0)
\\&\phantom{={}}
-vUm_0D_0(A_3v^*G_{-1}^0+A_2v^*G_{0}^0+A_1v^*G_1^0+v^*G_{2}^0).
\end{split}
\label{13.3.24.22.16}
\end{align}
Next, we use the operator identities
\begin{align*}
Sm_0=0,\quad Tm_1=0,\quad TM_0=0
\end{align*}
and their adjoints, cf.\ (\ref{12.12.27.14.14}), and then we can further simplify (\ref{13.3.24.22.16}):
\begin{align}
\begin{split}
HG_0
&=
1_{\mathcal L}+G_0^0vD_0v^*G_0^0
\\&\phantom{={}}
+vU\Bigl[Q-m_0B_0-m_0B_1C_0-m_1C_0-m_0B_2D_0
\\&\phantom{={}+vU\Bigl[}
-m_0B_1C_1D_0-m_1C_1D_0-m_2D_0\Bigr]
(A_1v^*G_{-1}^0+v^*G_0^{0}).
\end{split}\label{13.8.5.16.52}
\end{align}
Now we substitute the expressions (\ref{13.3.24.19.55}) and (\ref{13.8.5.13.41})
to the terms in the square brackets of (\ref{13.8.5.16.52}),
and then we can verify, after some computations, that they actually cancel out.
Since $G_0^0vD_0v^*G_0^0=-G_{-2}$, finally it follows that
\begin{align*}
HG_0=1_{\mathcal L}-G_{-2}.
\end{align*}
The adjoint computation verifies $G_0H=1_{\mathcal L}-G_{-2}$,
and hence we obtain the following theorem:
\begin{theorem}\label{13.3.11.4.32}
Under the assumption of Theorem~\ref{12.12.1.12.41} with $\beta\ge 5$
the coefficient $G_0$ satisfies
\begin{align*}
HG_0=G_0H=1_{\mathcal L}-G_{-2}.
\end{align*}
\end{theorem}
\begin{remark}
A remark similar to Remarks~\ref{13.8.18.5.25} 
hold also for this case.
We note that in this case we always have 
$\widetilde{\mathcal E}_{\mathrm{qs}}\neq \{0\}$.
\end{remark}

\appendix

\section{The threshold $\lambda=4$}\label{13.3.23.11.49}
The discrete Schr\"{o}dinger operators considered in the paper 
have thresholds at both $\lambda=0$ and $\lambda=4$. There is a simple relation between 
the two thresholds that makes it possible to reduce the analysis at 
threshold $\lambda=4$ to the one given for threshold $\lambda=0$. 
We define for any sequence $x$
\begin{equation*}
(Jx)[n]=(-1)^nx[n].
\end{equation*}
We see that $J$ satisfies $J^2=1$
and is a bounded self-adjoint and unitary operator 
on $\mathcal H$. 
Furthermore, if a potential $V$ satisfies 
Assumption~\ref{12.11.8.1.9}, then so does
the conjugation $V_J=JVJ^{-1}$. 
A straightforward computation yields for such $V$
\begin{equation}\label{reflect}
J(H_0+V)J^{-1}=-(H_0-V_J-4).
\end{equation}
Thus for the resolvents we have
\begin{equation*}
J(H_0+V-z)^{-1}J^{-1}=-(H_0-V_J-(4-z))^{-1}.
\end{equation*}
We note that for a multiplicative potential $V$ we have $V_J=V$.

\section{Examples}\label{12.11.8.2.35}

\subsection{Local and non-local potentials}\label{localV}

We provide a prototype of the potentials satisfying Assumption~\ref{12.11.8.1.9}.
This obviously includes \textit{non-local} operators of finite ranks.
\begin{proposition}\label{13.3.23.21.6}
Let $\beta\ge 1$ be any real number, and 
$v_j\in \mathcal L^\beta$, $j=1,2,\dots$, 
be at most a countable number of linearly independent vectors with
\begin{align}
\sum_{j}\|v_j\|_{\mathcal L^\beta}^2<\infty.
\label{13.3.23.21.7}
\end{align}
Then for any $\sigma_j\in \{\pm 1\}$ the operator series
\begin{align*}
V=\sum_{j}\sigma_j\langle v_j,\cdot\rangle v_j
\end{align*}
converge in the uniform topology of 
$\mathcal B((\mathcal L^\beta)^*,\mathcal L^\beta)$
and satisfy Assumption~\ref{12.11.8.1.9} with the same $\beta$. 
\end{proposition}
The proof is omitted.
We just have to take an abstract Hilbert space $\mathcal K$
with the same cardinality of a complete orthonormal system as $\#\{v_j\}$.
A potential $V$ given in Proposition~\ref{13.3.23.21.6} is 
\textit{local} in the following sense 
if and only if each $v_j$ has support consisting of one point.

\begin{definition}
A symmetric operator $V$ on the Hilbert space $\mathcal H=\ell^2(\mathbb Z)$
is said to be local, if it does not extend supports of functions.
\end{definition}

Since our base space is discrete, a local operator $V$
is always identified with a multiplication operator by some function,
which we shall denote also by $V$.
In particular, a local operator $V$ satisfies Assumption~\ref{12.11.8.1.9}, if 
\begin{align*}
\sum_{n\in\mathbb Z}(1+n^2)^{\beta}|V[n]|<\infty.
\end{align*}

\subsection{Local potentials with $P$ invertible}\label{local-P-inv}
With the following computations we show that for $V\neq0$ local only Case iii. in Proposition~\ref{12.12.19.6.30} occurs. If $V$ is local and of rank one, then up to a translation we have for some $c\neq0$ that
\begin{equation}
V=c\langle e_0,\cdot\rangle e_0,
\end{equation}
where $e_0[n]=1$ for $n=0$, and $e_0[n]=0$ for $n\neq0$.
The factorization is given as follows. We have $\mathcal{K}=\mathbb{C}$, $vz=|c|^{1/2}ze_0$, $z\in\mathcal{K}$, $v^{\ast}x=|c|^{1/2}\langle e_0,x\rangle$, $x\in (\mathcal{L}^s)^{\ast}$, any $s>0$, and $Uz=\mysign c \cdot z$.

Straightforward computations show that a basis for the solutions to $H\Psi=0$ is given by $\{u_1,u_2\}$, where
\begin{align*}
u_1[n]&=\begin{cases}
1, & n\leq0\\
1+cn, & n>0.
\end{cases}\\
u_2[n]&=\begin{cases}
1-cn, & n<0\\
1, & n\geq0.
\end{cases}
\end{align*}
To compare with the results in Table~\ref{13.3.6.0.26} we rewrite the basis as follows. We omit the details, which are straightforward.
\begin{align*}
u_1+u_2&=2\mathbf{1}+c|\mathbf{n}|=(2\mysign c)\Psi_5,\\
u_1-u_2&=c\mathbf{n}=\Psi_2^0.
\end{align*}

\subsection{Construction of threshold resonances}
Here we give examples of discrete Schr\"{o}dinger operator 
with a threshold resonance at $\lambda=0$. 
The technique for finding such examples is very simple. 
We are looking for a multiplicative 
potential $V$ such that there is a sequence 
$x\in(\mathcal L^0)^*$ satisfying
\begin{equation*}
-(x[n+1]+x[n-1]-2x[n])+V[n]x[n]=0.
\end{equation*}
We find such $V$ by first choosing $x$ and then taking
$V$ accordingly to 
\begin{equation*}
V[n]=\frac{x[n+1]+x[n-1]}{x[n]}-2.
\end{equation*}
Here are some examples of this technique.
\begin{examples}
We can take, e.g.,
\begin{enumerate}
\item 
\begin{equation*}
x[n]=\begin{cases}
2 & \mbox{if }n=0,\\
1 & \mbox{otherwise},
\end{cases}
\qquad
V[n]=\begin{cases}
-1 &\mbox{if } n=0,\\
1 &\mbox{if }  n=\pm 1,\\
0 &\mbox{otherwise};
\end{cases}
\end{equation*}
\item
\begin{equation*}
x[n]=\begin{cases}
3 &\mbox{if } n=0,\\
2 &\mbox{if }  n=\pm 1,\\
1 &\mbox{otherwise},
\end{cases}
\qquad
V[n]=\begin{cases}
-2/3 &\mbox{if } n=0,\\
0 &\mbox{if } n=\pm 1,\\
1 &\mbox{if } n=\pm 2,\\
0 &\mbox{otherwise.}
\end{cases}
\end{equation*}
\end{enumerate}
\end{examples}
Using the results from Appendix~\ref{13.3.23.11.49},
we can also get examples of threshold resonances at the threshold 
$\lambda=4$.
In fact, if we can construct nontrivial 
$x\in(\mathcal L^0)^*$ satisfying
$(H_0-V_J)x=0$,
then \eqref{reflect} implies that $y=J^{-1}x$ satisfies
$$
(H_0+V)y=4y.
$$
We note that these constructions work also in the continuous setting.

\subsection{Construction of threshold eigenvalues}\label{13.3.23.21.47}

As stated in Proposition~\ref{13.3.25.13.0}, 
for a multiplicative potential
there does not exist a decaying eigenfunction at the threshold.
However, we can construct an example of a non-local potential
that possesses linearly independent such
eigenfunctions as finitely many as we want.
We construct it using Proposition~\ref{13.3.23.21.6} in 
such way that $M_0$ has a nontrivial kernel, cf.\ Corollary~\ref{13.1.16.2.51}.

\begin{example}
Let us define the potential $V$ by 
\begin{align*}
V=-\sum_{j=1}^N \langle v_j,\cdot \rangle v_j;
\quad
v_j[n]=\left\{
\begin{array}{ll}
\sqrt2&\mbox{if }n=3j,\\
-1/\sqrt2&\mbox{if }n=3j\pm 1,\\
0&\mbox{otherwise.}
\end{array}
\right.
\end{align*}
Then the linearly independent sequences 
$\Psi_j\in\mathcal L$, $j=1,\dots,N$, given by
\begin{align*}
\Psi_j[n]
=\left\{
\begin{array}{ll}
1&\mbox{if }n=3j,\\
0&\mbox{otherwise}
\end{array}
\right.
\end{align*}
are obviously the decaying eigenfunctions for $H=H_0+V$.
\end{example}

\subsection{Construction of both threshold eigenvalues and resonances}
We now give an example exhibiting the simultaneous occurrence of threshold eigenvalues and resonances (an exceptional point of the third kind). 

\begin{example}
We define the following sequences:
\begin{align*}
\phi_j[n]&=\begin{cases}
-1, & n=4j,\\
\phantom{-}1, & n=4j+1,\\
\phantom{-}0, & \text{otherwise,}
\end{cases}
& j&=0,1,2,
\\
u_j[n]&=\begin{cases}
\phantom{-}1 & n\leq 4j,\\
-1 & n>4j,
\end{cases}
& j&=0,1,2.
\end{align*}
Then we define
\begin{equation*}
Vx=-\sum_{j=0}^2\langle \phi_j, x\rangle \phi_j.
\end{equation*}
With these definitions it is straightforward to verify that
\begin{equation*}
(H_0+V)u_j=0,\quad j=0,1,2.
\end{equation*}
If we define
\begin{equation*}
w[n]=\begin{cases}
1 & n=1,2,3,4,\\
0 & \text{otherwise},
\end{cases}
\end{equation*}
then we have $u_2=\frac12 u_0 + \frac12 u_1 +w$.
\end{example}
This example should be compared with the discussion in the introduction. It is easy to see that the example can be modified to provide a threshold eigenvalue of any finite multiplicity, besides the two linearly independent resonance functions.

\section*{Acknowledgement}
This work was completed while AJ visited the Department of Mathematics, Gakushuin University, Tokyo, Japan. AJ thanks the Department for its support and hospitality.


\begin{thebibliography}{10}

\bibitem{BMS}
A. Boutet de Monvel and J. Sahbani, \emph{On the spectral properties of discrete Schr\"{o}dinger operators},
C. R. Acad. Sci. Paris S\'{e}r. I Math. \textbf{328} (1999), no. 5, 443--448.

\bibitem{C}
S. Cuccagna, \emph{$L^p$ continuity of wave operators in $\mathbb{Z}$},
J. Math. Anal. Appl. \textbf{354} (2009), 594--605.

\bibitem{JK}
A.~Jensen and T.~Kato, \emph{Spectral properties of {S}chr\"{o}dinger operators
  and time-decay of the wave functions}, Duke Math. J. \textbf{46} (1979),
  583--611.


\bibitem{JN0}
A.~Jensen and G.~Nenciu, \emph{A unified approach to resolvent expansions at
  thresholds}, Rev. Math. Phys. \textbf{13} (2001), no.~6, 717--754.

\bibitem{JN1}
A.~Jensen and G.~Nenciu, \emph{Erratum: ``{A} unified approach to resolvent expansions at
  thresholds'' [{R}ev. {M}ath. {P}hys. {\bf 13} (2001), no. 6, 717--754]}, Rev.
  Math. Phys. \textbf{16} (2004), no.~5, 675--677.
  
  \bibitem{M}
M.~Murata, \emph{Asymptotic expansions in time for solutions of
  {S}chr\"odinger-type equations}, J. Funct. Anal. \textbf{49} (1982), no.~1,
  10--56.
  
  
  \bibitem{PS}
  D. E. Pelinovsky and A. Stefanov, \emph{On the spectral theory and dispersive estimates
for a discrete Schr\"{o}dinger equation in one dimension}, J. Math. Phys. \textbf{49} (2008), 113501.
  
  \bibitem{rauch}
J.~Rauch, \emph{Perturbation theory for eigenvalues and resonances of
  {S}chr\"o\-din\-ger {H}amiltonians}, J. Funct. Anal. \textbf{35} (1980), no.~3,
  304--315.
  
 
\end{thebibliography}
\end{document}